\documentclass[a4paper,USenglish]{article}

\usepackage[left=2cm,right=2cm,top=2cm,bottom=2.5cm]{geometry}
\usepackage{tikz}
\usetikzlibrary{quotes,graphs}
\usepackage{amsmath,amsthm,amssymb}
\usepackage[pdfborder={0 0 0}, hidelinks]{hyperref}
\usepackage{thm-restate}

\usepackage[textsize=scriptsize,obeyFinal]{todonotes}
\usepackage{babel}
\usepackage{
  amsfonts,
  wrapfig,
  multirow,
  pdflscape,
  booktabs
}

\usepackage{rotating}
\usepackage{multirow}

\usetikzlibrary{graphs}

\usetikzlibrary{calc}

\definecolor{ba.yellow}{RGB}{252,190,18}
\definecolor{ba.gray}{RGB}{153,153,156}
\definecolor{ba.blue}{RGB}{6,123,164}
\definecolor{ba.red}{RGB}{213,96,98}
\definecolor{ba.orange}{RGB}{233,116,81}
\definecolor{ba.pine}{RGB}{67,154,134}
\definecolor{ba.green}{RGB}{0, 168, 107}
\definecolor{ba.violet}{RGB}{88, 53, 94}

\newcommand{\annotate}[3]{\hspace{-#1}\vbox to 0pt{\llap{\raisebox{-#2}{\tikz{\draw[semithick, color=ba.red, |->] (0,0) to (0,-.5) node[below] {#3};} \ }}}\hspace{#1}}

\usepackage{mathtools}

\newtheorem{theorem}{Theorem}[section]
\newtheorem{corollary}[theorem]{Corollary}
\newtheorem{lemma}[theorem]{Lemma}

\theoremstyle{plain}
\newtheorem{example}[theorem]{Example}

\newtheorem{claim}{Claim}
\newtheorem{observation}{Observation}
\newtheorem*{claim*}{Claim}
\newtheorem*{corollary*}{Corollary}

\theoremstyle{definition}
\newtheorem*{notation*}{Notation}
\newtheorem{reductionrule}{Rule}
\newtheorem{branchingrule}{Branching Rule}

\newcommand{\edgenumber}{\ensuremath{(2d^{d}2^{2^d}\alpha)^{2^d-1}}}

\DeclareMathOperator{\clause}{clauses}
\DeclareMathOperator{\var}{vars}
\DeclareMathOperator{\link}{link}

\newcommand\Para{\mathrm{para\text-}}
\newcommand\Class[1]{\mathchoice{\text{\normalfont\small$\mathrm{#1}$}}{\text{\normalfont\small$\mathrm{#1}$}}{\text{\normalfont$\mathrm{#1}$}}{\text{\normalfont$\mathrm{#1}$}}}       
\newcommand{\Lang}[1]{\ifmmode{\text{\normalfont\textsc{#1}}}\else\normalfont\textsc{#1}\fi}
\newcommand{\PLang}[1][]{\mathrm p\def\test{#1}\ifx\test\stockhustantauempty\else_{#1}\fi\text-\penalty15\Lang}

\theoremstyle{plain}
\newtheorem{problem}[theorem]{Problem}
\newenvironment{parameterizedproblem}%
{%
  \leavevmode\nobreak\par
  \begin{list}%
    {}%
    {%
      \def\labelstyle{\itshape}
      \setlength{\topsep}{0pt}%
      \settowidth{\labelwidth}{\labelstyle Parameter:}%
      \setlength{\leftmargin}{\labelwidth}%
      \addtolength{\leftmargin}{\labelsep}%
      \setlength{\itemsep}{0pt}%
      \setlength{\parsep}{0pt}%
    }%
      \def\instance{\item[\labelstyle Instance:]}%
      \def\parameter{\item[\labelstyle Parameter:]}%
      \def\question{\item[\labelstyle Question:]}%
    }%
    {%
    \end{list}%
}

\tikzset{
  vertex/.style = {
    draw,
    fill,
    semithick,
    circle,
    inner sep     = 0pt,
    minimum width = 1.5mm,
  },
  mirror/.style = {
    vertex,      
    minimum width = 0.75mm,
    draw = ba.violet,
    fill = ba.violet!50
  },
  edge/.style = {
    semithick
  },
  mirroredge/.style = {
    semithick
  },
}
\def\nodeset#1#2{
  \node[vertex, label=$a^{\scalebox{.6}{#2}}_{\scalebox{.6}{#1}}$] (a) at (0,0)  {};
  \node[vertex, label=$b^{\scalebox{.6}{#2}}_{\scalebox{.6}{#1}}$] (b) at (0,-1) {};
  \node[vertex, label=$c^{\scalebox{.6}{#2}}_{\scalebox{.6}{#1}}$] (c) at (1,0)  {};
  \node[vertex, label=$d^{\scalebox{.6}{#2}}_{\scalebox{.6}{#1}}$] (d) at (2,0)  {};
  \node[vertex, label=$e^{\scalebox{.6}{#2}}_{\scalebox{.6}{#1}}$] (e) at (2,-1) {};
  \node[vertex, label=$f^{\scalebox{.6}{#2}}_{\scalebox{.6}{#1}}$] (f) at (3,-1) {};
  \node[vertex, label=$g^{\scalebox{.6}{#2}}_{\scalebox{.6}{#1}}$] (g) at (3,0)  {}; 
}
\def\edgeset{
  \graph[use existing nodes, edges = {edge}]{
    a -- b -- c -- d -- e -- f -- g -- d;
    a -- c;
  };
}
\def\mirrors{
  \node[mirror] (a') at ($(a)+(0.25,0.25)$) {};
  \node[mirror] (b') at ($(b)+(0.125,0.25)$) {};
  \node[mirror] (c') at ($(c)+(0.25,0.25)$) {};
  \node[mirror] (d') at ($(d)+(0.25,0.25)$) {};
  \node[mirror] (e') at ($(e)+(0.25,0.25)$) {};
  \node[mirror] (f') at ($(f)+(0.25,0.25)$) {};
  \node[mirror] (g') at ($(g)+(0.25,0.25)$) {};
  \graph[use existing nodes, edges = {mirroredge}]{
    a' -- a;
    b' -- b;
    c' -- c;            
    d' -- d;
    e' -- e;
    f' -- f;
    g' -- g;            
  };
}

\title{MaxSAT with Absolute Value Functions:\\ A Parameterized Perspective}
\author{%
  Max Bannach%
  \thanks{Institute for Theoretical Computer Science, Universit{\"a}t zu L{\"u}beck, Germany (\href{mailto:bannach@tcs.uni-luebeck.de}{\texttt{bannach@tcs.uni-luebeck.de}})}%
  \and%
  Pamela Fleischmann%
  \thanks{Department of Computer Science, Kiel University, Germany (\href{mailto:fpa@informatik.uni-kiel.de}{\texttt{fpa@informatik.uni-kiel.de}})}%
  \and%
  Malte Skambath%
  \thanks{Department of Computer Science, Kiel University, Germany (\href{mailto:malte.skambath@email.uni-kiel.de}{\texttt{malte.skambath@email.uni-kiel.de}})}%
}
\date{\today}

\begin{document}

\maketitle

\begin{abstract}
The natural generalization of the Boolean satisfiability problem to
optimization problems is the task of determining the maximum number of
clauses that can simultaneously be satisfied in a propositional
formula in conjunctive normal form. In the \emph{weighted maximum satisfiability problem}
each clause has a positive weight and 
one seeks an assignment of maximum weight. 
The literature almost solely considers the case of
\emph{positive} weights. While the general case of the problem is only
restricted slightly by this constraint, many special cases become
trivial in the absence of \emph{negative} weights. In this work we 
study the problem \emph{with} negative weights and observe that the
problem becomes computationally harder~--~which we formalize from a parameterized
perspective in the sense that various variations of the problem become $\Class{W[1]}$-hard if
negative weights are present.

Allowing negative weights also introduces new variants of the
problem: Instead of maximizing the sum of weights of satisfied
clauses, we can maximize the \emph{absolute value} of that sum. This
turns out to be surprisingly expressive even restricted to \emph{monotone} formulas in
\emph{disjunctive normal form} with at most \emph{two literals} per
clause. In contrast to the versions without the absolute value,
however, we prove that these variants are \emph{fixed-parameter tractable}.
As technical contribution we present a kernelization for an auxiliary problem on hypergraphs
in which we seek, given an edge-weighted hypergraph, an induced
subgraph that maximizes the absolute value of the sum of edge-weights. 
\end{abstract}

\clearpage

\section{Introduction}

Maximum satisfiability is the natural generalization of the
satisfiability problem of propositional logic to optimization
problems. In its decision version $\Lang{max-sat}$, we seek a truth
assignment for a given propositional formula in conjunctive normal
form that satisfies at least $\alpha\in\mathbb{N}$ clauses.  This
problem is known to be $\Class{NP}$-complete even in the restricted
case that every clause contains at most two
literals~\cite{GareyJS76}. On the positive side, it is known that the
problem can be decided in time $f(\alpha)\cdot n^{c}$ for some
function~$f\colon\mathbb{N}\to\mathbb{N}$ and constant
$c\in\mathbb{N}$, which means that the problem is
\emph{fixed-parameter tractable} parameterized by~$\alpha$~\cite{MahajanR99}. This
also holds for the weighted version in which clauses have positive
integer weights. In this paper we extend the previous scenario by allowing also \emph{negative weights}. We
denote the corresponding problem with $\Lang{max-cnf}$: decide for a given propositional formula with 
\emph{arbitrary integer} weights whether there is an
assignment of weight at least $\alpha$. In the presence of negative
weights, many non-trivial variants of the problem arise, for
instance $\Lang{max-dnf}$ (the propositional formula is in disjunctive
normal form), $\Lang{max-monotone-cnf}$ and $\Lang{max-monotone-dnf}$
(all literals in the formulas are positive), and $\Lang{abs-cnf}$ as
well as $\Lang{abs-dnf}$ (maximize the \emph{absolute value} of the
sum of the satisfied clauses). Even an obscure combination such as
$\Lang{abs-monotone-dnf}$ becomes surprisingly expressive if negative
weights are allowed (we show that this version is $\Class{NP}$-hard
restricted to formulas in which every clause contains at most two literals).
This complexity jump (many of these versions are trivially in~$\Class{P}$ if negative weights are \emph{not} allowed) motivates a
systematic study of the parameterized complexity of $\Lang{abs-dnf}$
and its variants. The main problem we investigate is defined as:

\begin{problem}[{$\PLang[\alpha]{$d$-abs-dnf}$}]
  \begin{parameterizedproblem}
    \instance A propositional formula $\phi$ in disjunctive normal
    form with clauses of size at most $d$, a target value $\alpha\in\mathbb{N}$, and a weight function
    $w\colon\clause(\phi)\rightarrow\mathbb{Z}$.
    \parameter $\alpha$
    \question Is there an assignment $\beta$ such that
    $|\sum_{c\in\clause(\phi),\beta\models c}w(c)|\geq\alpha$ (the absolute value of the sum of the weights of
    satisfied clauses is at least $\alpha$)?
  \end{parameterizedproblem}
\end{problem}

We investigate the complexity of this problem by varying the status of
$\alpha$ and $d$ as parameter. In the notation above ($\alpha$ as
index of $p$ and $d$ in the problem description) we consider
$\alpha$ as a parameter and $d$ as a fixed constant. But we also
consider $\alpha$ \emph{and} $d$ to be parameters at the same time, or
both to be constant, or to be neither a parameter nor a constant.
An overview of the complexity of $\Lang{abs-dnf}$ can be found in Table~\ref{table:results}.

\begin{table}[h]
  \caption{Summary of our results for the maximum satisfiability
    problem over formulas in disjunctive normal form with negative
    weights and an absolute value target function (\Lang{abs-dnf}). We
    distinguish how $\alpha$ (the weight of the sought solution) and $d$
    (maximum size of a clause) are interpreted (as \emph{constant}, as
    \emph{parameter}, or as \emph{unbounded}). The
    columns correspond to~$d$, while the rows
    represent~$\alpha$. Upper bounds are highlighted in
    \textcolor{ba.green!65!black}{green} and lower bounds in
    \textcolor{ba.red}{red}.}
  \label{table:results}
  \centering
  \begin{tabular}{r@{\hskip 3ex}ccc}
    %
      $\alpha$~\raisebox{-.5ex}{\tikz\draw[semithick](1.5ex,-1ex)--(-1.5ex,1ex);}~$d$ 
    & \emph{constant}
    & \emph{parameter}
    & \emph{unbounded}
    \\[2ex]
    %
    \emph{constant}
    & \parbox{3cm}{\centering\color{ba.green!65!black}$\Class{P}$\\[-1mm]{\tiny (Corollary~\ref{corollary:abs-dnf-p})}} 
    & \parbox{3cm}{\centering\color{ba.green!65!black}$\Class{FPT}$\\[-1mm]{\tiny (Corollary~\ref{corollary:abs-dnf})}} 
    & unknown
    \\[2ex]
    %
      \emph{parameter}
    & \parbox{3cm}{\centering\color{ba.green!65!black}$\Class{FPT}$\\[-1mm]{\tiny (Theorem~\ref{theorem:abs-dnf})}} 
    & \parbox{3cm}{\centering\color{ba.green!65!black}$\Class{FPT}$\\[-1mm]{\tiny (Corollary~\ref{corollary:abs-dnf})}} 
    & \parbox{3cm}{\centering\color{ba.red}$\Class{W}[1]$-hard\\[-1mm]{\tiny (Corollary~\ref{corollary:abs-dnf-without-d})}} 
    \\[2ex]
    %
      \emph{unbounded}
    & \parbox{3cm}{\centering\color{ba.red}$\Class{NP}$-hard\\[-1mm]{\tiny (Corollary~\ref{corollary:abs-dnf-np})}} 
    & \parbox{3cm}{\centering\color{ba.red}$\Para\Class{NP}$-hard\\[-1mm]{\tiny (Corollary~\ref{corollary:abs-dnf-paranp})}} 
    & \parbox{3cm}{\centering\color{ba.red}$\Class{NP}$-hard\\[-1mm]{\tiny (Corollary~\ref{corollary:abs-dnf-np})}} 
    \\
  \end{tabular}
\end{table}

As mentioned before, we will also study other variations of the
problem, which we obtain by replacing the $\textsc{dnf}$ by
a $\textsc{cnf}$, by forcing the formula to be monotone, or by replacing
the absolute value function by a normal sum. It turns out that without
the absolute value function, many natural versions of the problem become
$\Class{W[1]}$-hard in the presence of negative weights. 
However, the problems become tractable if the aim
is to maximize the absolute value of the target sum as they than can
all be fpt-reduced to
$\PLang[\alpha]{$d$-abs-monotone-dnf}$, which in turn is equivalent to
the following hypergraph problem:

\begin{problem}[{$\PLang[\alpha]{$d$-unbalanced-subgraph}$}]\label{problem:unbalancedsubgraph}
  \begin{parameterizedproblem}
    \instance A $d$-hypergraph $H=(V,E)$, a weight function $w\colon
    E\to\mathbb{Z}$, and a number~$\alpha\in\mathbb{N}$.
    \parameter $\alpha$
    \question Is there a set $X\subseteq V$ with $|w[X]|\geq \alpha$.
  \end{parameterizedproblem}
\end{problem}

In words, we seek a subset $X$ of the vertices of an edge-weighted
hypergraph such that the absolute value of $w[X]=\sum_{e\in E[X]}
w(e)$ is maximized. Here $E[X]$ are the hyperedges induced by the set $X$. Intuitively, that means we have to find a subgraph
that contains either more positive than negative edges or more
negative edges than positive ones. Thus, we seek an {\em unbalanced} subgraph.

\subparagraph*{Our Contribution I: Hardness Results.}  We present
hardness results for the parameterized and non-parameterized variants
of $\Lang{max-dnf}$ and $\Lang{abs-dnf}$ (both with negative
weights). In particular, we prove that $\Lang{max-dnf}$ and
$\Lang{abs-dnf}$ are $\Class{NP}$-hard when restricted to monotone
formulas with at most two literals per clause. However, the problems
differ from a parameterized perspective:
$\PLang[\alpha]{$d$-max-monotone-dnf}$ is $\Class{W}[1]$-hard for all
$d\geq 2$, while $\PLang[\alpha]{$d$-abs-dnf}\in\Class{FPT}$.
We also prove a similar result for the \textsc{cnf}-versions: $\PLang[\alpha]{$d$-max-monotone-cnf}$ 
is $\Class{W}[1]$-hard for all
$d\geq 2$, but $\PLang[\alpha]{$d$-abs-cnf}\in\Class{FPT}$.

\subparagraph*{Our Contribution II: Satisfiability Based Optimization with Absolute Value Function.}
We derive a kernelization algorithm for
$\PLang[\alpha]{$d$-unbalanced-subgraph}$, which implies that the
problem is in $\Class{FPT}$ parameterized by $\alpha$. Using an
fpt-reduction we show:

\begin{restatable}{theorem}{absdnf}\label{theorem:abs-dnf}%
  $\PLang[\alpha]{$d$-abs-dnf}\in\Class{FPT}$
\end{restatable}

\subparagraph*{Our Contribution III: Absolute Integer Optimization.}
Theorem~\ref{theorem:abs-dnf} can be generalized to an
\emph{absolute integer optimization problem.} Given a multivariate polynomial
$p\colon \mathbb{Z}^n\to\mathbb{Z}$ with integer coefficients and
intervals~$[(b_{\textit{min}})_{i};(b_{\textit{max}})_{i}]$ for every $i\in\{1,\dots,n\}$, the
task is to find a solution $(x_1,\dots,x_n)\in\mathbb{Z}^n$ with
$(b_{\textit{min}})_{i}\leq x_i \leq (b_{\textit{max}})_{i}$ such that the
non-linear objective function $|p(x_1,\dots,x_n)|$ is maximized. 
With the value of the sought solution $\alpha$ and the maximal degree $d$ of the polynomial,
we prove:

\begin{restatable}{theorem}{absip}\label{theorem:abs-ip}%
  $\PLang[\alpha, d]{abs-io}\in\Class{FPT}$
\end{restatable}

We emphasize that there are no restrictions on the number of variables
or on the size of an interval
$[(b_{\textit{min}})_{i};(b_{\textit{max}})_{i}]$. There may be even
infinitely large intervals like $[-\infty;\infty]$.

\subparagraph*{Related Work.}  Optimization problems based on
satisfiability or constraint-satisfaction problems are usually
$\Class{NP}$-complete~\cite{KhannaSW97, Schaefer78}. The parameterized
complexity of these problems is, thus, an active field of
research~\cite{GaspersS11, GaspersS14, Grohe06, Szeider11}. An
overview over the current trends can be found in~\cite{DellKLMM17,
  SamerS10}.  In a propositional formula in conjunctive normal form,
we can always satisfy at least $(1-2^{-d})m$ clauses if~$m$ is the
number of clauses and $d$ is the size of these
clauses~\cite[Chapter~9.2]{CyganFKLMPPS15}. Therefore, the problem is
fixed-parameter tractable when parameterized by the solution
size~$\alpha$. It is known that the problem remains in
$\Class{FPT}$ parameterized above this lower bound, i.\,e., the
parameter is the distance between $(1-2^{-d})m$ and the number of
clauses that can actually be satisfied
simultaneously~\cite{AlonGKSY11}. The currently fastest algorithm for
parameter $\alpha$ runs in time $O^*(1.3248^{\alpha})$ and improving
the base is an active field of research~\cite{ChenXW17}.

Another parameterized approach to (non-optimizing) satisfiability
problems are \emph{backdoors}, which are small sets of
variables such that assigning values to these variables yields an easy
formula, e.\,g., a \textsc{krom}-
or \textsc{horn}-formula~\cite{CramaEH97,
  NishimuraRS04}. Unfortunately, \Lang{max-cnf} remains
$\Class{NP}$-hard on such formulas, making the approach less appealing
for optimization problems.

A third line of research on parameterized algorithms for
satisfiability problems focuses on \emph{structural parameters},
i.\,e., graph theoretic properties of a graph representation of the
formula. Depending on the details of that representation, various
tractability and intractability results can be derived for parameters such as
the \emph{treewidth} of the corresponding graph~\cite{SamerS10b,
  SamerS10}. Such algorithms usually can be obtained from general
frameworks such as Courcelle's Theorem and carry over
to the optimization versions~\cite{Courcelle90}. Generalizing this
approach to a broader range of graphs is an active field of research,
see for instance~\cite{GanianS17} or~\cite[Chapter~17]{HoS}.

A related problem is
\Lang{max-lin2}~\cite{AlonGKSY11,CrowstonFGJRTY2011,MahajanRS2009}:
given variables $x_1,\dots,x_n$ and a set of
equations $\Pi_{i\in I}x_i=b$ for $b\in\{-1,1\}$ and
$I\subseteq\{1,\dots,n\}$, find an assignment to
$\{-1,1\}$ that satisfies as many equations as possible. This can 
be seen as maximizing \emph{parity constraints}, in contrast to
\emph{or constraints} (\Lang{max-cnf}) or \emph{and constraints} (\Lang{max-dnf}).
The parameterized complexity of \emph{linear} integer programming is
well understood, see for instance the famous result of
Lenstra~\cite{Lenstra1983} or recent developments in algorithms based
on structural parameters~\cite{EibenGKO18, GanianO19,
  JansenLR20, KouteckLO18}. There are also recent results on \emph{non-linear}
programs, see for instance~\cite{EibenGKO2019, GavenciakKK2019,
  GaveniakKK2020}. For separable convex functions,
fixed-parameter tractability  results are known for structural parameters depending on the constraint
matrix~\cite{EisenbrandHKKLO19}. There are also results that directly
cover absolute values of polynomials~--~for instance, for maximizing convex functions,
which include absolute values of convex polynomials~\cite{GaveniakKK2020}.

\subparagraph*{Structure of this Work} After some brief preliminaries
in Section~\ref{section:preliminaries}, we present lower bounds for
various versions of $\Lang{abs-dnf}$ and related problems in
Section~\ref{section:hardness}. Afterwards, we formulate our main
result~--~an $\Class{FPT}$-algorithm for
$\PLang[\alpha]{$d$-abs-dnf}$~--~in Section~\ref{section:absdnf}. We
also present the underlying machinery, i.\,e., a
kernelization-algorithm for
$\PLang[\alpha]{$d$-unbalanced-subgraph}$. Section~\ref{section:absip} concludes with further
applications  and an outlook is presented in
Section~\ref{section:conclusion}.

\section{Preliminaries}\label{section:preliminaries}

A \emph{parameterized problem} is a set
$Q\subseteq\Sigma^*\times\mathbb{N}$, where an instance $(w,k)$
consists of a word~$w$ and a \emph{parameter} $k$. We denote
parameterized problems with a preceding ``$\mathrm{p}$-'' with an index
that indicates what the parameter is~--~for instance,
$\PLang[k]{vertex-cover}$ denotes the well-known vertex cover problem
parameterized by the solution size~$k$. A parameterized problem is
\emph{fixed-parameter tractable} (\emph{fpt} or it is in $\Class{FPT}$) if there is
a computable function $f\colon\mathbb{N}\rightarrow\mathbb{N}$ and a
constant $c\in\mathbb{N}$ such that we can decide $(w,k)\in^? Q$ in time $f(k)\cdot |w|^c$.
A \emph{size-$g$ kernelization} for a parameterized
problem $Q$ and a computable function
$g\colon\mathbb{N}\rightarrow\mathbb{N}$ is a polynomial-time
algorithm that, on input of an instance~$(w,k)$, outputs an instance
$(w',k')$ such that $(w,k)\in Q\Leftrightarrow
(w',k')\in Q$ and $|w'|+k'\leq g(k)$. The output of the algorithm is
called a \emph{kernel}. If $g$ is a polynomial then the kernel is called a \emph{polynomial kernel}. It is well-known that a decidable problem is in $\Class{FPT}$
iff it admits a kernelization~\cite[Theorem~1.39]{FlumG06}.

A \emph{$d$-hypergraph}~$H=(V,E)$ consists of a set $V$ of vertices
and a set $E$ of edges with $e\subseteq V$ and $|e|\leq d$ for all
$e\in E$. A hypergraph is \emph{$d$-uniform} if all edges have size
\emph{exactly}~$d$. For $d=2$, they are just \emph{graphs}. The \emph{neighborhood} of a vertex $v$ is the set
$N(v)=\{\,w\mid (\exists e\in E)(v,w\in e)\,\}$ and its
degree~$\deg(v)$ is the number of edges containing it. The \emph{link}
of a set $c\subseteq V$ is defined as $\link(c)=\{\,e\mid \text{$e\in
  E$ and $c\subsetneq e$}\,\}$.
A \emph{subedge} of some hyperedge $e\in E$ is any subset~$s\subseteq e$.
We use $\Delta(H)=\max_{v\in V}\deg(v)$ (or
just $\Delta$, if $H$ is clear from the context) to denote the
\emph{maximum degree of $H$}.  For $X\subseteq V$ and $w\colon E\to\mathbb{Z}$, let
$E[X]=\{\,e\in E\mid e\subseteq X\,\}$ denote the edges of the
\emph{induced hypergraph} $H[X]=(X,E[X])$ over $X$, and define its
\emph{weight} as $w[X]=\sum_{e\in E[X]} w(e)$. 

We need only little terminology from propositional logic. A
formula $\phi$ is in \emph{disjunctive
normal form} if it is a disjunction of conjunctions, e.\,g.,
$\psi\equiv(x_1\wedge x_2)\vee(x_2\wedge x_3)$. It is  in \emph{conjunctive normal form} if it
is a conjunction of disjunctions, e.\,g., $\chi\equiv(x_1\vee
x_2)\wedge(x_2\vee \neg x_3)$. We
refer to the variables of $\phi$ with $\var(\phi)$, and we call
the conjunctions in a formula in disjunctive normal form and the
disjunctions in a formula in conjunctive normal form
\emph{clauses}. These objects are addressed with $\clause(\phi)$. For
instance, $\var(\psi)=\var(\chi)=\{x_1,x_2,x_3\}$,
$\clause(\psi)=\{(x_1,x_2),(x_2,x_3)\}$, and
$\clause(\chi)=\{(x_1,x_2),(x_2,\neg x_3)\}$. A \emph{literal} is a
variable (then it is \emph{positive}) or its negation (then it is
\emph{negative}). Formulas that contain only positive literals are
\emph{monotone}, e.\,g., $\psi$ is monotone but $\chi$ is not.

\section{Lower Bounds for MaxSAT and AbsSAT Based Optimization}\label{section:hardness}

The problems $\Lang{max-dnf}$ and $\Lang{abs-dnf}$ may seem similar on
first sight. After all, if there are no negative weights they are
clearly equivalent. However, if negative weights are present, a
solution of $\Lang{abs-dnf}$ may either construct a sum of at least
$\alpha$ or of at most $-\alpha$; in contrast, a solution for
$\Lang{max-dnf}$ does not have such freedom and has to find a
weight-$\alpha$ solution. This results in an interesting complexity
gap: $\PLang[\alpha]{$d$-max-monotone-dnf}$ is $\Class{W}[1]$-hard for
$d\geq 2$ while $\PLang[\alpha]{abs-dnf}$ is only $\Class{W}[1]$-hard
if $d$ is unbounded (recall that $d$ is the size of the clauses and
that in $\PLang[\alpha]{$d$-max-monotone-dnf}$ the size of every
clause is bounded by a constant $d$). We prove the hardness results in
this section and complement them with a proof of
$\PLang[\alpha]{$d$-abs-dnf}\in\Class{FPT}$ in the next section.

We first consider $\PLang[\alpha]{$d$-max-monotone-dnf}$. This problem
is $\Class{W}[1]$-hard by a reduction from the independent set problem
(given a graph and an integer $k$, decide whether there is a set of at
least $k$ pairwise non-adjacent vertices):

\begin{lemma}\label{lemma:max-dnf}
  $\PLang[\alpha]{$d$-max-monotone-dnf}$ is $\Class{W}[1]$-hard for $d\geq 2$.
\end{lemma}
\begin{proof}
  Let $G=(V,E)$ and $k\in\mathbb{N}$ be an instance of
  $\PLang[k]{independent-set}$. We build a formula $\phi$ that
  contains a variable $x_v$ for every $v\in V$:
  \[
    \phi\equiv
    \bigvee_{v\in V} (x_v)\annotate{-0.15cm}{1.2cm}{1}
    \;\vee
    \bigvee_{\{v,w\}\in E}(x_v\wedge x_w)\annotate{0.25cm}{1.2cm}{-1}.
  \]
  The red labeling shows the corresponding weight function
  $w$. Clearly, a size-$k$ independent set translates into an
  assignment of weight $k$. For the other direction consider an
  assignment of weight~$k$. We can assume that the assignment
  satisfies only positively weighted clauses. Otherwise we can flip
  one variable of a satisfied and negatively weighted clause. This
  does not decrease the weight because there is at most one positively
  weighted clause containing this variable. Since $k$ variables
  are set to \emph{true} and no negatively weighted clauses are satisfied,
  the assignment translates back to an independent set.
\end{proof}

\begin{corollary}\label{corollary:max-dnf-w1}
    $\PLang[\alpha]{$d$-max-dnf}$ is $\Class{W}[1]$-hard for $d\geq 2$.
\end{corollary}

On the other hand, the absolute value version of the problem remains
intractable in the classical sense, which motivates our parameterized
perspective. In detail, the (unparameterized) problem
$\Lang{$d$-abs-dnf}$ is $\Class{NP}$-hard by a reduction from
$\Lang{independent-set}$. Note that in contrast to the previous lemma,
the reduction in the following is not parameter-preserving. Thus, the
theorem does \emph{not} imply that
$\PLang[\alpha]{$d$-abs-monotone-dnf}$ is $\Class{W}[1]$-hard.

\begin{theorem}\label{theorem:abs-dnf-np}
  $\Lang{$d$-abs-monotone-dnf}$ is $\Class{NP}$-hard for $d\geq 2$.
\end{theorem}
\begin{proof}
  We reduce from $\Lang{independent-set}$ and let $G=(V,E)$ with
  $k\in\mathbb{N}$ be a corresponding instance. Construct the
  following formula $\phi$ that contains four
  variables~$v_1^+,v_2^+,v_1^-,v_2^-$ for every vertex $v\in V$ (the
  red arrows illustrate the weight function):
  \[
    \phi\equiv
    \bigvee_{v\in V}
      (v_1^+\wedge v_2^+)\annotate{0.25cm}{1.2cm}{-1}
      \vee
      (v_1^-\wedge v_2^-)\annotate{0.25cm}{1.2cm}{1}
    \;\vee
    \bigvee_{\{v,w\}\in E}
      (v_1^+\wedge w_1^+)\annotate{0.27cm}{1.2cm}{1}
      \vee
      (v_1^-\wedge w_1^-)\annotate{0.27cm}{1.2cm}{-1}.
  \]
  The intuition (also illustrated in Example~\ref{example:reduction})
  behind the formula is as follows: Think of two weighted copies of
  $G$ called $G^+$ and $G^-$ such that all edges in $G^+$ have weight
  $+1$ and all edges in $G^-$ have weight $-1$. This is the second
  part of the formula. Observe that by taking either all edges of
  $G^+$ or all edges of $G^-$, a solution would always have absolute
  value $|E|$.

  To encode the independent set problem, we introduce an extra gain
  for every vertex, these are the variables $v_2^+$ and
  $v_2^-$. Observe that in the first part of $\phi$, we can add $-1$
  or $+1$ to the solution by setting $v_1^{+/-}$ and $v_2^{+/-}$ to
  \emph{true} at the same time. Finally, note that the signs in the
  first part of the formula are exactly the opposite of the signs in
  the second part.  We claim that $G$ has a size-$k$ independent set
  iff $\phi$ has a solution of weight $k+|E|$. Let $S$ be a size-$k$
  independent set in $G$, then the following assignment has an
  absolute value of $k+|E|$:
  \[
    \beta(v_i^{\sigma}) =
    \begin{cases}
      \mathit{true}  & \text{if ($v\in S$ and $\sigma=-$) or ($\sigma=+$ and $i=1$); } \\
      \mathit{false} & \text{else.                      }
    \end{cases}
  \]
  For the other direction let $\beta$ be an assignment that achieves
  an absolute value $|\nu|$ of at least~$k+|E|$. Assume for simplicity
  that $\nu>0$ (the case $\nu<0$ is symmetric). We may assume
  $\beta(v_1^+)=\mathit{true}$ and
  $\beta(v_2^+)=\mathit{false}$~--~otherwise we can increase the
  solution by modifying the assignment locally. Observe that by
  setting only variables $v_i^+$ to $\textit{true}$, we have
  $\nu\leq|E|$ and, since we assumed $|\nu|\geq k+|E|$, $\beta$ sets
  further variables~$v_j^-$ to $\textit{true}$.  We may assume
  $\beta(v_1^-)=\beta(v_2^-)$ for all $v\in V$, as setting both
  variables to $\textit{true}$ is the only way to increase $\nu$ (the
  clauses of the form $(v_1^-\wedge w_1^-)$ all have a negative weight
  and we assumed $\nu>0$).

  Let $\{v,w\}\in E$ be an edge such that
  $\beta(v_1^-)=\beta(w_1^-)=\textit{true}$. Observe that changing
  $\beta(w_1^-)$ to $\beta(w_1^-)=\textit{false}$ does \emph{not}
  decrease the value of the solution (it decreases the number of
  satisfied clauses with negative weight by exactly one, as the
  positively weighted clause $(w_1^-\wedge w_2^-)$ is not satisfied
  anymore; but it increases the number of satisfied clauses with
  positive weight by at least one as the negatively weighted clause
  $(v_1^-\wedge w_1^-)$ is also no longer satisfied). Hence, we may
  adapt $\beta$ such that $S=\{\,v\mid \beta(v_1^-)=\textit{true}\,\}$
  is an independent set. Since the absolute value of $\beta$ is at
  least $|E|+k$ and since all clauses containing variables of the form
  $v_i^+$ can contribute at most~$|E|$ to this sum, we have
  $|S|\geq k$.
\end{proof}

\begin{example}\label{example:reduction}
  Let us illustrate the reduction used within the theorem with an
  example. Consider the following graph $G$ with an
  \textcolor{ba.blue}{independent set} of size three:

  
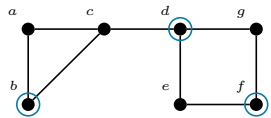
{
    \centering
    
    \begin{tikzpicture}[
      every node/.style={font=\tiny, label position=-225, label distance=-1mm},
      independent/.style = {
        draw = ba.blue,
        circle,
        semithick,
        inner sep = 0pt,
        minimum width = 3mm
      }
      ]
      \nodeset{}{}
      \edgeset
      \node[independent] at (b) {};
      \node[independent] at (d) {};
      \node[independent] at (f) {};
    \end{tikzpicture}
    \medskip

  }

  Within the reduction, we virtually generate the following
  edge-weighted graph that
  contains two copies of $G$, namely
  \textcolor{ba.green!65!black}{$G^+$} and
  \textcolor{ba.red}{$G^-$}. In both copies, every vertex has a
  \textcolor{ba.violet}{mirror} vertex attached to it. Edges have
  either a \textcolor{ba.green!65!black}{positive} weight of
  \textcolor{ba.green!65!black}{+1} or a \textcolor{ba.red}{negative}
  weight of \textcolor{ba.red}{-1}:

  {
    \centering
    \medskip
    
    \begin{tikzpicture}[
      every node/.style={font=\tiny, label position=-225, label distance=-1mm},
      vertex/.append style = {draw=ba.green!65!black, fill=ba.green!50},
      edge/.append style = {color=ba.green!65!black},
      mirroredge/.append style = {color=ba.red}
      ]
      \nodeset{1}{+}
      \edgeset
      \mirrors
    \end{tikzpicture}
    \qquad
    \begin{tikzpicture}[
      every node/.style={font=\tiny, label position=-225, label distance=-1mm},
      vertex/.append style = {draw=ba.red, fill=ba.red!50},
      edge/.append style = {color=ba.red},
      mirroredge/.append style = {color=ba.green!65!black}
      ]
      \nodeset{1}{-}
      \edgeset
      \mirrors            
    \end{tikzpicture}

  }

  The reduction continues by defining a weighted formula that contains
  a clause for every edge with the same weight as that edge, e.\,g., we
  would have \textcolor{ba.green!65!black}{$(a_1^+\wedge b_1^+)$} of
  weight \textcolor{ba.green!65!black}{$+1$} and
  \textcolor{ba.red}{$(a_1^-\wedge b_1^-)$} of weight
  \textcolor{ba.red}{$-1$}. By selecting from one copy (i.\,e., setting the corresponding
  variables to \emph{true}) all vertices except the mirrors, we obtain an assignment of weight $\pm|E|$;
  selecting the independent set together with its mirrors in the other
  copy provides another $\pm k$, yielding a solution of absolute
  value $|E|+k$. In this example we could take all green vertices
  (providing a weight of $+8$) and \textcolor{ba.blue}{$\{b_1^-,
    d_1^-, f_1^-\}$} together with their mirrors (for another $+3$);
  yielding a solution with absolute value $11$.
  \hfill$\lrcorner$
\end{example}

\begin{corollary}\label{corollary:abs-dnf-np}
  $\Lang{$d$-abs-dnf}$ is $\Class{NP}$-hard for $d\geq 2$.
\end{corollary}
\begin{corollary}\label{corollary:abs-dnf-paranp}
  $\PLang[d]{abs-dnf}$ is $\Para\Class{NP}$-hard.
\end{corollary}

Due to the previous theorem, we cannot hope to achieve parameterized tractability with respect to
the sole parameter $d$. The following lemma shows that this is also
not possible for the parameter $\alpha$. In the upcoming sections,
we thus rely on the combined parameter $\alpha+d$.

\begin{lemma}\label{lemma:abs-monotone-dnf-without-d}
    $\PLang[\alpha]{abs-monotone-dnf}$ is $\Class{W[1]}$-hard.
\end{lemma}
\begin{proof}
    Let $G=(V,E)$ and $k\in\mathbb{N}$ be an instance of
    $\PLang[k]{independent-set}$. Consider the following formula~$\phi$
    that contains a variable~$x_v$ for every vertex~$v\in V$:    
    \[
      \phi \equiv \bigvee_{v\in V}
      \left(
        \neg x_v \wedge \bigwedge_{\{v,w\}\in E} x_w
      \right).
    \]
    
    Let $\alpha=k$ and $w$ be a weight function that
    maps all clauses of $\phi$ to $1$. We show that $(G,k)$ is a
    yes-instance of $\Lang{independent-set}$ iff
    $(\phi,w,\alpha)$ is a yes-instance of $\Lang{abs-dnf}$.
    For the first direction let $I\subseteq V$ be a size-$k$ independent
    set and consider the assignment:
    \[
      \beta(x_v) =
      \begin{cases}
        \mathit{false}  & \mbox{if }v\in I; \\
        \mathit{true} & \text{otherwise.} \\
      \end{cases}
    \]

    Since $I$ is an
    independent set we have $\beta(x_w)=\mathit{true}$ for all
    vertices $w$ that have a neighbor $v\in I$. Therefore, $\beta$ satisfies every
    clause that corresponds to a vertex of $I$ and, hence, $\beta$ is
    an assignment of weight at least $k$.
  
    For the other direction let $\beta$ be an assignment of weight at
    least $k$, i.\,e., one that satisfies at least $k$ clauses. For
    each of these clauses there is a unique vertex $v$ such that $x_v$
    appears negatively only in that clause and, thus, at least $k$
    variables must be assigned to \emph{false} by $\beta$. Furthermore, for
    all neighbors $w$ of $v$ we have $\beta(x_w)=\mathit{true}$ by
    the second part of the clause. We conclude that
    $\{\,v\mid\beta(x_v)=\mathit{false}\,\}$ is an independent set of
    size at least $k$.

    The formula $\phi$ is not monotone, however, in the proof of
    Lemma~\ref{lemma:abs-dnf-to-abs-monotone-dnf} we describe how an
    arbitrary instance can be turned into an equivalent one with a
    monotone formula. In this case we obtain the following monotone
    formula~$\phi'$, which (instead of $\phi$) can be constructed
    directly for this reduction.
    \[
      \phi' \equiv \bigvee_{v\in V}
        \left(\bigwedge_{\{v,w\}\in E} x_w \right) \annotate{0.75cm}{1.6cm}{1}
      \;\vee\; \bigvee_{v\in V}
      \left(x_v \wedge \bigwedge_{\{v,w\}\in E} x_w 
      \right)\annotate{1cm}{1.6cm}{-1}.
    \]
    The red labeling shows the corresponding weight function~$w'$. As
    we describe in the proof of
    Lemma~\ref{lemma:abs-dnf-to-abs-monotone-dnf} the
    instances~$(\phi,w,k)$ and~$(\phi',w',k)$ are equivalent. This
    completes the proof.
   \end{proof}

   \begin{corollary}\label{corollary:abs-dnf-without-d}
    $\PLang[\alpha]{abs-dnf}$ is $\Class{W[1]}$-hard.
   \end{corollary}
  
\subsection{From Disjunctive Normal Form to Conjunctive Normal Form}
So far, we have seen that the parameterized $\textsc{dnf}$ maximization
variants are $\Class{W[1]}$-hard. We can use this to show that the same
holds also for the conjunctive normal forms:
   
\begin{lemma}\label{lemma:max-monotone-cnf-to-dnf}
  $\PLang[\alpha]{$d$-max-monotone-cnf}$ is $\Class{W}[1]$-hard for $d\geq 2$.
\end{lemma}
\begin{proof}
  We reduce from $\PLang[\alpha]{$d$-max-monotone-dnf}$ and let
  $(\phi,w,\alpha)$ be the corresponding instance. We may assume that
  $\phi$ does not contain the empty clause (which is always satisfied) as
  the reduction in Lemma~\ref{lemma:max-dnf} does not produce such clauses.

  The idea of the reduction is as follows.
  We initialize a set of clauses $\Phi$ as $\Phi=\clause(\phi)$ and
  then, as long as $\Phi$ contains a conjunction $c=(x_1
  \wedge x_2\wedge \dots\wedge x_{\ell})$, we update $\Phi$ to:
  \[
    \Phi = \big(
    \Phi\setminus\{c\}
    \big)
    \cup
    \big\{\,
    \bigvee_{v\in X}v\mid \emptyset\neq X\subseteq\var(c)
    \,\big\}.
  \]

  Let us denote the newly added clause for the set
  $X\subseteq\var(c)$ with $c_X$. We set the weight of $w(c_X)$ to  $(-1)^{|X|+1} w(c)$.   
  Since every conjunction is replaced by at most $2^d-1$ new
  disjunctions, the reduction can be carried out in time $2^d\cdot
  |\clause(\phi)|^t$  for some constant~$t\in\mathbb{N}$.

  To prove the correctness, let $(\Phi, w)$ be a set of
  clauses containing at least one conjunction together with a
  weight function $w$, and let $(\Phi', w')$ be the result of the
  just described operation for some conjunction $c$. We show that any truth assignment
  $\beta$ satisfies:
  \begin{equation}\label{eqn:conj-to-disj-equiv}
    \sum_{\sigma\in\Phi, \beta\models \sigma} w(\sigma) = \sum_{\sigma\in\Phi', \beta\models \sigma} w'(\sigma).
  \end{equation}
  We identify truth values with $0$ and $1$ and show by induction that for any $s>0$ and
  any set of binary values~$x_1,\dots,x_s\in\{0,1\}$ we have:
  \begin{equation}\label{eqn:binary-multi-product-to-max-term}
    x_1\cdot x_2 \cdots x_{s-1}\cdot x_s 
    = \hspace{1ex}\sum_{\mathclap{\emptyset\neq I\subseteq \{1,\dots,s\}}}\hspace{1ex} (-1)^{|I|+1}\max \{\, x_i\mid i\in I \,\}. 
  \end{equation}
  Since $x=\max\{x\}$, the equation holds for $s=1$. For $s=2$ and 
  $x,y\in\{0,1\}$ we observe $\max\{x,y\}=x+y-xy$,
  which implies 
  \begin{equation}\label{eqn:binary-product-to-max-term}
    xy=-\max\{x,y\}+\max\{x\}+\max\{y\}.
  \end{equation}
  Let us assume the equation holds for any set of at most $s-1$
  values. Then by induction, we can rewrite the product as
  \[
    x_1\cdot x_2 \cdots x_{s-1}\cdot x_s 
    = \hspace{1ex}\sum_{\mathclap{\emptyset\neq I\subseteq \{1,\dots,s-1\}}}\hspace{1ex} (-1)^{|I|+1}x_s\max \{\, x_i\mid i\in I\,\}.
  \]
  Using (\ref{eqn:binary-product-to-max-term}) we derive
  \[
    x_1\cdot x_2 \cdot \cdots \cdot x_{s-1}\cdot x_s 
    = \hspace{1ex}\sum_{\mathclap{\emptyset\neq I\subseteq \{1,\dots,s-1\}}}\hspace{1ex}
    (-1)^{|I|+1}(x_s+\max \{ x_i\mid i\in I \}-\max \{\, x_i\mid i\in I\cup\{s\} \,\})
  \]
  and by rearranging the sums we obtain:
  \begin{align*}
    x_1\cdot x_2 \cdots x_{s-1}\cdot x_s &                                                               \\
    =\quad &\sum_{\emptyset\neq I\subseteq \{1,\dots,s-1\}} (-1)^{|I|+1} x_s                                     \\
    +\quad &\sum_{\emptyset\neq I\subseteq \{1,\dots,s-1\}} (-1)^{|I|+1} \max \{\, x_i\mid i\in I \,\}           \\
    +\quad &\sum_{\emptyset\neq  I\subseteq \{1,\dots,s-1\}} (-1)^{|I|+2} \max \{\, x_i\mid i\in I\cup\{s\} \,\}  \\
    \\
    =\quad &\sum_{\emptyset\neq I\subseteq \{1,\dots,s-1\}} (-1)^{|I|+1} x_s                                      \\
    +\quad &\sum_{\substack{\emptyset\neq  I\subseteq \{1,\dots,s\}\\I\neq \{x_s\}}} (-1)^{|I|+1} \max \{ x_i\mid i\in I \}.
  \end{align*}
  Applying the binomial theorem yields:
  \begin{align*}
    x_1\cdot x_2 \cdot \cdots \cdot x_{s-1}\cdot x_s & \\
    =\quad &x_s + \hspace{1ex}\sum_{\mathclap{\substack{\emptyset\neq I\subseteq \{1,\dots,s\}\\I\neq \{x_s\}}}}\hspace{1ex} (-1)^{|I|+1} \max \{\,x_i\mid i\in I\,\} \\
    =\quad &\phantom{x_s \mathrel{+} }\hspace{1ex}\sum_{\mathclap{\emptyset\neq I\subseteq \{1,\dots,s\}}}\hspace{1ex} (-1)^{|I|+1} \max \{\, x_i\mid i\in I\,\}. \\
  \end{align*}
  
  We derive equation~(\ref{eqn:conj-to-disj-equiv}) as follows:
  \begin{align*}
    &\sum_{\mathclap{\sigma\in\Phi, \beta\models\sigma}} w(\sigma)\\
    =\qquad &\sum_{\mathclap{\sigma\in\Phi\setminus\{c\}, \beta\models\sigma}} w(\sigma) &+&\quad  &&\phantom{w(c)\cdot}\sum_{\mathclap{\sigma\in\{c\},\beta\models\sigma}}\quad w(\sigma)\hspace{6cm}  \\
    =\qquad &\sum_{\mathclap{\sigma\in\Phi\setminus\{c\}, \beta\models\sigma}} w(\sigma) &+&\quad  &&w(c)\cdot\prod_{\mathclap{x_i\in\var(\sigma)}}\quad\beta(x_i)\tag{see above}\\
    =\qquad &\sum_{\mathclap{\sigma\in\Phi\setminus\{c\}, \beta\models\sigma}} w(\sigma) &+&\quad  &&w(c)\cdot\sum_{\mathclap{\emptyset\neq X\subseteq \var(c)}}\quad(-1)^{|X|+1}\max\{\,\beta(x)\mid x\in X\,\}\tag{by (\ref{eqn:binary-multi-product-to-max-term})}\\
    =\qquad &\sum_{\mathclap{\sigma\in\Phi\setminus\{c\}, \beta\models\sigma}} w'(\sigma) &+&\quad  &&\phantom{w(c)\cdot}\sum_{\mathclap{\emptyset\neq X\subseteq \var(c)}}\quad w'(c_X)\max\{\,\beta(x)\mid x\in X\,\}\tag{Def. $c_X$}\\
    =\qquad &\sum_{\mathclap{\sigma\in\Phi', \beta\models\sigma}} w'(\sigma).\tag*{\qedhere}
  \end{align*}
\end{proof}

\begin{corollary}
  $\PLang[\alpha]{$d$-max-cnf}$ is $\Class{W[1]}$-hard for $d\geq 2$.
\end{corollary}

\subsection{Finding an Assignment with a Certain Weight}
We finish this section with a variant of the problem in which we seek
an assignment of a certain absolute value. In detail,
$\PLang[\alpha,d]{exact-abs-monotone-dnf}$ asks, given a weighted
propositional formula in disjunctive normal form, whether there is an
assignment such that the absolute value of the weighted sum is
\emph{exactly}~$\alpha$ (and not \emph{at least~$\alpha$}).

\begin{lemma}\label{lemma:exact-abs-monotone-dnf-np}
  $\Lang{$d$-exact-abs-monotone-dnf}$ is $\Class{NP}$-hard for any constant~$d\geq 2$ and $\alpha = 0$.
\end{lemma}
\begin{proof}
  Since $\PLang[\alpha]{$d$-max-monotone-dnf}$ is $\Class{W[1]}$-hard
  by Lemma~\ref{lemma:max-dnf}, it is not surprising that the problem
  $\Lang{$d$-exact-max-monotone-dnf}$ is $\Class{NP}$-hard. We just
  use the same polynomial-time-computable reduction, which is possible
  since the independent set problem is monotone in the sense that the
  existence of a solution of size $\geq k$ implies the existence of an
  independent set of size $=k$.
    
  We can turn a $\Lang{$d$-exact-max-monotone-dnf}$ instance into a
  \Lang{$d$-exact-abs-monotone-dnf} instance: Just add the empty
  clause with weight $-\alpha$ and set the target value to $0$. Every
  assignment of weight $\alpha$ in the original instance will have
  weight $0$, as the empty clause is always satisfied. On the other
  hand, any weight-$0$ assignment in the new formula has to satisfy
  clauses of weight exactly $\alpha$ to compensate for the empty
  clause.
\end{proof}

\begin{corollary}\label{corollary:exact-absdnf-para-np}
  $\PLang[\alpha]{$d$-exact-abs-monotone-dnf}$ is $\Para\Class{NP}$-hard.
\end{corollary}
\begin{proof}
  A parameterized problem is $\Para\Class{NP}$-hard if there are
  finitely many slices whose union is $\Class{NP}$-hard~\cite[Theorem
  2.14]{FlumG06}. The claim follows as already the slice $\alpha=0$
  alone is $\Class{NP}$-hard by
  Lemma~\ref{lemma:exact-abs-monotone-dnf-np}.
\end{proof}


\section{Satisfiability Based Optimization with Absolute Value Function}\label{section:absdnf}

In the previous section we showed that various \textsc{sat}-based
optimization problems are presumably not in $\Class{FPT}$ with respect
to the natural parameter~$\alpha$. There are various ways to deal with
these negative results. For instance, we could restrict ourselves to
structured instances~--~such as formulas of bounded
incidence-treewidth~\cite{DellKLMM17}. For the absolute version of the
problem, we fortunately do not need such a strong structural
parameter. Instead, we provide a reduction to an auxiliary hypergraph
problem for which we present a polynomial kernel. This results in the
main theorem of this section:

\absdnf*

It is well-known that various optimization problems over \textsc{cnf}s
and \textsc{dnf}s are equal to problems on hypergraphs. For instance,
it is easy to see that the question of whether a given \emph{monotone}
\textsc{$d$-cnf} can be satisfied by setting at most $k$ variables to
\emph{true} equals the hitting set problem on $d$-hypergraphs. Notice that
``monotone'' here is a crucial key word, as this property allows us to
encode clauses as edges. Such an approach does \emph{not} work
directly for $\PLang[\alpha]{$d$-max-dnf}$. Fortunately,
$\PLang[\alpha]{$d$-abs-dnf}$ reduces to its monotone version:

\begin{lemma}\label{lemma:abs-dnf-to-abs-monotone-dnf}
  $\PLang[\alpha]{$d$-abs-dnf}\leq\PLang[\alpha]{$d$-abs-monotone-dnf}$
\end{lemma}
\begin{proof}
  Let $\phi$ be a \textsc{$d$-dnf} and
  $w\colon\clause(\phi)\rightarrow\mathbb{N}$ be a weight function on
  its clauses. We argue how we can remove one occurrence of a negative
  literal~--~an iterated application of this idea then leads to the
  claim.

  Let $c\in\clause(\phi)$ be a clause containing a negative literal
  $\overline x$. We remove $c$ from $\phi$ and add two new clauses:
  $c^+$ and $c^-$. The clause $c^+$ equals $c$ without $\overline x$,
  and we set $c^-=c^+\wedge x$. Define $w(c^+)=w(c)$ and
  $w(c^-)=-w(c)$, and observe that an assignment that satisfies~$c$
  does satisfy $c^+$ but not $c^-$ and, thus, yields the same
  value. On the other hand, every assignment that satisfies $c^-$ does
  also satisfy $c^+$ and, hence, the values cancel.
\end{proof}

Given a monotone \textsc{dnf} and a weight function on its clauses, we
encode $\PLang[\alpha]{$d$-abs-monotone-dnf}$ as a
hypergraph by introducing one vertex per variable and by representing
clauses as edges. The task then is to identify a set of vertices such
that the absolute value of the total weight of the edges appearing in
the subgraph induced by the vertices is as large as possible. Formally
this is Problem~\ref{problem:unbalancedsubgraph} from the
introduction.

\begin{corollary}\label{corollary:dnf-to-unbalanced}
  $\PLang[\alpha]{$d$-abs-dnf}\leq \PLang[\alpha]{$d$-unbalanced-subgraph}$
\end{corollary}

\subsection{A Polynomial Kernel for Unbalanced Subgraph}

In this section we develop a polynomial kernel for
$\PLang[\alpha]{$d$-unbalanced-subgraph}$. The kernelization consists
of a collection of \emph{reduction rules} which we just call
\emph{rules.} Every rule obtains as input an instance of
$\PLang[\alpha]{$d$-unbalanced-subgraph}$ and either (i)~is not
applicable, or (ii)~outputs an equivalent smaller instance. The rules
are numbered and we assume that whenever we invoke a rule, all rules
with smaller numbers have already been applied exhaustively.  It will
always be rather easy to show that the rules can be implemented in
polynomial time, so we only prove that they are safe. To get started,
we consider the following simple rules, where $\Delta$ is the maximum
vertex degree in the hypergraph:

\begin{reductionrule}\label{rule:isolated-vertices}
  If $H$ contains an isolated vertex $v$, delete $v$.
\end{reductionrule}

\begin{reductionrule}\label{rule:edges-without-weights}
  If $H$ contains an edge $e$ with $w(e)=0$, delete $e$.
\end{reductionrule}

\begin{reductionrule}\label{rule:degree-rule}
  If $H$ has at least $2\alpha\cdot d^3\Delta^2$ vertices,
  reduce to a trivial yes-instance.
\end{reductionrule}

\begin{lemma}\label{lemma:rule-three-safe}
  Rule~\ref{rule:isolated-vertices}--\ref{rule:degree-rule} are safe.
\end{lemma}
\begin{proof}
  For the first two rules observe that neither isolated vertices nor
  weight-$0$ edges can contribute to any solution and, hence, may be
  discarded.
   
  For Rule~\ref{rule:degree-rule} we argue that $H$ contains a set
  $X\subseteq V$ with $|w[X]|\geq\alpha$. Consider a maximal set
  $M\subseteq E$ such that (i) all edges in $M$ are pairwise disjoint
  and (ii) $H\big[\cup_{e\in M}e\big]=\big(\cup_{e\in M}e,M\big)$,
  that is, $M$ is a set packing that induces itself. Such a set always
  exists and can also be computed as the following greedy algorithm
  shows: Start with $M=\emptyset$ and repeatedly pick a remaining
  nonempty edge~$e$ from~$H$ that is no superset of any edge in $H$
  (except $\emptyset$) and add it to $M$. Since
  Rule~\ref{rule:isolated-vertices} has already removed isolated
  vertices, there must be such an edge as long as $V$ is not
  empty. After picking an edge $e$, remove the set~$E_e$ of all edges
  that intersect $e$. Observe that $|E_e|\leq d\Delta$. Clearly, if we
  just remove~$E_e$, the algorithm would already generate a maximal
  set packing~--~however, it would not necessarily be
  self-induced. Let therefore
  $V_e=\{\,v\in V\mid(\exists e'\in E_e)(v\in e')\,\}$ be the set of
  vertices covered by the to-be-removed edges. Since every edge has
  size at most $d$, we have $|V_e|\leq d^2\Delta$. To ensure that $M$
  will be self-induced, we remove all edges that contain a vertex of
  $V_e$~--~which are at most $d^2\Delta^2$ edges as $H$ has maximum
  degree~$\Delta$.

  For every edge that is picked to be in $M$, at most $d^2\Delta^2$
  edges are removed from $|E|$. Therefore we have
  $|M|\geq |E|/\big(d^2\Delta^2\big)$. Furthermore, we have
  $|E|\geq |V|/d$ because Rule~\ref{rule:isolated-vertices} has
  already removed isolated vertices. Given that
  $|V|\geq 2\alpha\cdot d^3\Delta^2$, we deduce $|M|\geq
  2\alpha$. There are at least $\alpha$ edges in $M$ that have the
  same sign (no edge has a zero-weight by
  Rule~\ref{rule:edges-without-weights}). Therefore, the set $X$ of
  vertices induced by the larger set satisfies $|w[X]|\geq\alpha$.
\end{proof}

The rules yield a kernel for the combined parameter
$\alpha+\Delta$, which means that an exhaustive application of these
rules leaves an instance of size bounded
in $\alpha$ and $\Delta$.

\begin{corollary}\label{corollary:degree-kernel}
  $\PLang[\alpha,\Delta]{$d$-unbalanced-subgraph}$ has a kernel with
  $O(\alpha\cdot\Delta^2)$ vertices.
\end{corollary}

If we only consider $\alpha+d$ as parameter, rules~1--3 are not
applicable, and the graph is still large, we know that the input graph
has no isolated vertices, all edges have a non-zero weight, and we
know that there is at least one vertex of high degree. Our final
reduction rule states that in this case we can as well reduce to a
trivial yes-instance, since the high degree vertex induces a
solution. Intuitively, this is similar to an application of the
well-known Sunflower Lemma~\cite{ErdosR1960}, which states that large
enough hypergraphs contain a set of hyperedges that pairwise all
intersect in the same set of vertices (i.\,e., the ``hypergraph
version'' of a high-degree vertex). Unfortunately, a direct
application of the lemma is tricky as other hyperedges may intersect
arbitrarily with such sunflowers. To circumnavigate this
difficulties, we need objects that are more like a \emph{self-induced
  sunflower}. This leads to the following reduction rule; details of
these concepts are postponed to the proof of
Lemma~\ref{lemma:rule-four-safe}.

\begin{reductionrule}\label{rule:subedge-rule}
  Let $g(i)=(i^{i}2\alpha\cdot 2^{2^d})^{2^i-1}$ for $i>0$ and
  $g(0)=1$. If $H$ contains a subedge~$c\subseteq e\in E$ such that
  (i) $|\link(c)|\geq g(i)$ where $i=d-|c|$ and (ii) for every
  superset~$f\supsetneq c$ it holds that $|\link(f)|< g(j)$ where
  $j=d-|f|$, then reduce to a trivial yes-instance.
\end{reductionrule}

Before we prove the safeness of the rule in the general case, let us
quickly sketch why the rule is correct for 2-uniform hypergraphs,
i.\,e., normal graphs. A vertex $v$ with $\deg(v)\geq
4\alpha$ is incident to at least $2\alpha$ edges of the same sign. For the
sake of an example let there be a set $P(v)\subseteq N(v)$ with $|P(v)|\geq2\alpha$
and $w(\{v,u\})>0$ for all $u\in P(v)$. Either $P(v)$ is a solution
(i.\,e., $|w[P(v)]|>\alpha$) or $P(v)\cup\{v\}$ is a solution (since
$w[P(v)\cup\{v\}]\geq 2\alpha-|w[P(v)]|\geq\alpha$).

\begin{lemma}\label{lemma:rule-four-safe}
  Rule~\ref{rule:subedge-rule} is safe.
\end{lemma}
\begin{proof}
  Let $c\subseteq V$ be a set for which the rule is applicable. We
  argue that $H$ contains a set $X$ with
  $|w[X]|\geq\alpha$.  Let $i=d-|c|$,
  $E(c)=\link(c)=\{e\in E\mid c\subsetneq e\}$,
  $N(c)=\bigcup E(c)\setminus c$, and
  $\deg_{E(c)}(v)=|\{e\in E(c)\mid v\in e\}|$. We can assume that
  $|c|<d$ and $i>0$ because otherwise $|\link(c)|=0 < g(i)$. Since
  $|\link(f)|< g(j)$ holds for any superset~$f\supsetneq c$ with
  $j=d-|f|$, it also holds that~$\deg_{E(c)}(v) \leq g(i-1)$ for every
  vertex~$v\in N(c)$.  Otherwise we have a contradicion with
  $f=c\cup\{v\}$.

  Consider a maximal set $M\subseteq E(c)$ such that (i) all edges in
  $M$ contain $c$ as a subedge, (ii) $c$ is the intersection of any
  pair of different edges in $M$, and (iii) for every edge~$e\in E(c)$
  it holds that either $e\in M$, or
  $e\not\subseteq \big(\bigcup_{e'\in M} e'\big)$, that is, $M$ is a
  set of edges that induces only itself in $H_c=(V,E(c))$. By (i) and
  (ii) $M$ is also known as a \emph{sunflower} with \emph{core}
  $c$. Such a set always exists as the following greedy algorithm
  shows: Start with $M=\emptyset$ and observe that all three
  properties hold. Repeatedly pick an edge $e\in E(c)$ that has no
  smaller subedge in $E(c)$ and add it to $M$ as long as the
  properties are preserved.

  We will use $M$ to identify the soughted set~$X$. For this, we
  require a lower bound for the size of $M$. Therefore, let us
  carefully analyze which edges cannot be choosen after an edge has
  been added to~$M$. After picking an edge $e\in E(c)$, we can remove
  the whole set $E_e \subseteq E(c)$ of all edges that intersect $e$
  not only in $c$ (a proper superset of $c$) because picking one of
  those edges contradicts property (ii).  There may be also an
  edge~$e'\in E(c)\setminus E_e$ that has not been added yet, but for
  which there is another edge $e''\in E(c)\setminus M$ with
  $e''\in \bigcup M \cup\{e'\}$. This implies that $e'$ may not be
  added to $M$ to preserve property (iii). Furthermore, there must be
  an edge $e\in M$ such that $e'$ intersects either $e$ or some edge in
  $E_e$ outside in the core. Let $E'_e\subseteq E(c)$ be the set of
  all edges that intersect $E_e$ not only in $c$. It follows that any
  edge in $E(c)$ will either be inserted to $M$, or is in $E_e$ or
  $E'_e$ for an edge $e\in M$. Observe that $|E_e|\leq ig(i-1)$
  because $e\setminus c$ contains at most $i$ vertices that are in at
  most $g(i-1)-1$ edges from $E(c)$. For the same reason, we have
  $|E'_e|\leq |E_e|\cdot (i-1)\cdot g(i-1)$. Therefore,
  $|\{e\}\cup E_e\cup E'_e|\leq i^2(g(i-1))^{2}$ edges are removed for
  every edge in~$M$. This implies
  $|M|\geq |E(c)|/\big(i^2 g(i-1)^2\big)$. Given that $E(c)\geq g(i)$,
  we deduce~$|M|\geq 2\alpha\cdot 2^{2^d}$: For $i=1$, we have
  $|M|\geq|E(c)|$ because of $g(0)=1$. For $i>1$ we use that
  $x^{2^i-1}/x^{2\cdot 2^{i-1}-2}= x$ and $g(i) = x^{2^i-1}$ for
  $x=2\alpha\cdot 2^{2^d}$.
  
  We show that we can find a set~$X\subseteq
  \bigcup_{e\in M}e$ with $|w[X]|\geq \alpha$.
  Let $M^+$ denote all positively weighted edges
  in $M$ and let $M^-$ denote the other ones. Let us consider
  three cases based on the size of the core.

  \subparagraph*{First Case ($|c|=0$).}
  We have $c=\emptyset$ and, without loss of generality, let
  $w[\emptyset]=|w[\emptyset]|\geq 0$. Otherwise we could flip the sign
  of all weights. Note that
  $E(\emptyset)=E\setminus\{\emptyset\}$. By Property~(iii), any
  subset of $\bigcup M$ induces only edges in $M\cup\{\emptyset\}$.
  It follows that $w[\bigcup M^+] = |M^+| + w[\emptyset]$.  We are
  done for $w[\bigcup M^+]|\geq\alpha$, otherwise we obtain
  \[
    |w[\bigcup M^-]| = |M^-| - w[\emptyset] = |M|-(|M^+|+w[\emptyset])>|M|-\alpha.
  \]
  The claim follows with $|M|\geq 2\alpha$.

  \subparagraph*{Second Case ($|c|=1$).}  In this case we have
  $c=\{v\}$ for some vertex $v\in V$. Without loss of generality, let
  $w[c]=|w[c]|\geq 0$~--~otherwise we can flip the sign of all weights
  again.
  In contrast to the previous case, a subset of $\bigcup M\cup c$
  might induce edges in~$M$ and subedges of
  $c$. Those edges are induced by $V' = \bigcup M \setminus c$. For
  any~$X'\subseteq V'$ we have~$|w[X']|<\alpha$, otherwise we are
  done. Observe that $w[\bigcup M^+] = |M^+| + w[c] + w[\bigcup
  M^+]$. The statement follows immediately for
  $w[\bigcup M^+]|\geq\alpha$. Otherwise we obtain
  \begin{align*}
    |M^+|+w[c]       & <\alpha+w[\bigcup M^+]<2\alpha\quad\text{and}\quad \\
    |w[\bigcup M^-]| & = |M^-| - w[c] + w[\bigcup M^-]                    \\
                     & = |M|-|M^+|-w[c] - w[\bigcup M^-] >|M|-3\alpha.
  \end{align*}
  The claim follows with $|M|\geq 4\alpha$.
  
  \subparagraph*{Third Case ($|c|\geq 2$).}  Observe that in this case
  $\bigcup M$ may induce edges that are neither in~$E(c)$ nor subedge
  of $c$, but that intersect $c$ (and other vertices covered by
  $M$). This is because $M$ contains only edges of $E(c)$.  Let us
  recursively define the following subsets of
  $E\left[\bigcup M^{\sigma}\right]$ for $\sigma\in\{-,+\}$:
  \begin{align*}
     E_{\emptyset}^{\sigma} 
       &= E\left[\left(\bigcup M^\sigma \setminus c\right)\right]
          \setminus E[c];\\
     E_{c'}^{\sigma} 
       &= E\left[c'\cup \left(\bigcup M^\sigma \setminus c\right)\right]
          \setminus \bigcup_{f\subsetneq c'} E_{f}^\sigma
          \setminus E[c]
          \quad\text{for every }c'\subseteq c.
  \end{align*}

  For all $c'\subseteq c$ these sets are disjoint and every
  edge~$e\in E[\bigcup M^{\sigma}]$ is either in one of these sets or
  in $E[c]$. To be precisely, an edge $e\in E[\bigcup M^{\sigma}]$ is
  in $E^{\sigma}_{e\cap c}$ if it is a proper superset of~$c$.
  Otherwise $e\subseteq c$ (if $e\subseteq c$). Note that
  $M^\sigma = E^\sigma_c$. We will now show that there is a set
  $c'\subseteq c$ and such that the absolute value of the induced
  weight of $c'\cup (\bigcup M^-\setminus c)$ or
  $c'\cup (\bigcup M^+\setminus c)$ is at least $\alpha$. Assume this
  is not the case and that for all subsets~$X\subseteq \bigcup M$ we
  have
  \begin{equation}\label{eqn:flower-eqn-1}
    |w[X]|<\alpha
  \end{equation}
  for any $c'\subseteq c$ and $\sigma\in\{-,+\}$. 
  Observe that
  \begin{equation}
    E\left[c' \cup \left(\bigcup M^\sigma\setminus c\right) \right] 
    = E^\sigma_{c'} \mathop{\dot\cup} E[c'] 
    \mathop{\dot\cup} 
    \bigcup_{f\subsetneq c'} E^\sigma_{f}.
  \end{equation}
   Let $w^\sigma_{c'}=\sum_{e\in E^\sigma_{c'}} w(e)$ and note that this is the total weight of the edges in
  $E^\sigma_{c'}$. Hence:
  \begin{equation}\label{eqn:flower-eqn-2}
    w\left[c' \cup \left(\bigcup M^\sigma\setminus c\right) \right] 
    = w^\sigma_{c'} + w[c'] 
    + \sum_{f\subsetneq c'} w^\sigma_{f}.
  \end{equation}

  For any~$k\in\{0,\dots,|c|\}$ let~$w_k=2^{2^k-1}\cdot 2\alpha$. We
  show that $w_k$ is an upper bound in the sense
  that~$|w^\sigma_{c'}| < w_k$ for $c'\subseteq c$ and~$k=|c'|$. By
  (\ref{eqn:flower-eqn-1}) and (\ref{eqn:flower-eqn-2}) we deduce that
  \[ 
    |w^\sigma_\emptyset| = \left| w\left[\left(\bigcup M^\sigma\setminus c\right) \right] - w[\emptyset] \right|< 2\alpha = w_0.
  \]
  Otherwise either $|w[\emptyset]|\geq\alpha$ or $\left|
  w\left[\left(\bigcup M^\sigma\setminus c\right)
  \right]\right|\geq\alpha$. Observe that $2\alpha + (2^{k-1}-1)
  w_{k-1}\leq w_k$ for $k>0$. Then for $\emptyset\neq c'\subseteq c$
  it follows by induction that
  \begin{align*}
    |w^\sigma_{c'}| 
      =& \left| w\left[c' \cup \left(\bigcup M^\sigma\setminus c\right) \right] 
         - w[c'] 
         - \sum_{f\subsetneq c'} w^\sigma_f
         \right|\\
       <&\; \alpha + \alpha + \sum_{k=0}^{|c'|-1}\binom{|c'|}{k} w_{k}\\
       \leq&\; 2\alpha + (2^{|c'|}-1)w_{|c'|-1} \\
       \leq&\; w_{|c'|}.
  \end{align*}
  
  Since
  $|M|\geq 2\alpha\cdot 2^{2^d}\geq 4\alpha\cdot 2^{2^{|c|}-1} =
  2\cdot w_{|c|}$ either $|M^+|\geq w_{|c|}$ or $|M^-|\geq w_{|c|}$.
  Recall that $M^\sigma = E^\sigma_c$ and, thus, with
  $|w^\sigma_c|\geq |M^\sigma|$ we get that our assumption must be
  wrong.
\end{proof}

\begin{corollary}\label{corollary:kernel}
  $\PLang[\alpha]{$d$-unbalanced-subgraph}$ has a kernel with $O\big(\alpha^{2^d}\big)$ vertices.
\end{corollary}
\begin{proof}
  If Rule~4 cannot be applied, there is no such
  subedge~$c\subseteq e\in E$ that meets the two properties of the
  rule. However, this also implies that there is no subedge~$c$ that
  only meets the first property because otherwise there would be an
  inclusionwise largest superset (at last a size-$(d-1)$ edge), which
  does. Therefore, for any subedge~$c\subseteq e\in E$ it holds that
  $|\link(c)|<(i^{i}2\alpha\cdot 2^{2^d})^{2^i-1}$ where
  $i=d-|c|$. Especially for $c=\emptyset$ we have
  $|\link(c)|<\edgenumber{}$. By the definition of the link we have
  $|E|\leq |\link(\emptyset)|+1$. Thus, the total number of edges is
  at most~$\edgenumber{}$.  Hence, we get a kernel of size
  $O\big(\alpha^{2^d}\big)$.
\end{proof}

\begin{proof}[Proof of Theorem~\ref{theorem:abs-dnf}]
  Combine Corollary~\ref{corollary:dnf-to-unbalanced} with Corollary~\ref{corollary:kernel}.
\end{proof}

If $d$ is not a constant but a parameter, then, by an exhaustive
application of the rules above, we obtain reduced instances of size
$O\big( d\edgenumber{}\big)$.  Note that there could be up to $n^d$
possible subedges, therefore it is not clear that we obtain a
kernel. However, a closer look at Rule~4 reveals that
we are not forced to consider these subedges~--~we can directly return a
trivial yes-instance if $|E|\geq \edgenumber$.

\begin{corollary}%
  \label{corollary:kernel-parameters-alpha-d}
  $\PLang[\alpha,d]{unbalanced-subgraph}\in\Class{FPT}$
\end{corollary} 
\begin{proof}
  An induction shows that when Rule~1 and Rule~2 were applied
  exhaustively and $|E|\geq \edgenumber{}$, Rule~$4$ can be
  applied on $c=\emptyset$ because
  $E\setminus\{\emptyset\}=\link(\emptyset)$ or on some larger
  subedge. Hence we can return a
  trivial yes-instance whenever $|E|\geq \edgenumber{}$.
\end{proof}

\begin{corollary}\label{corollary:abs-dnf}
  $\PLang[\alpha,d]{abs-dnf}\in\Class{FPT}$
\end{corollary}

\begin{corollary}\label{corollary:abs-dnf-p}
  $\Lang{$\alpha$-$d$-abs-dnf}\in\Class{P}$
\end{corollary}

\subsection{From Disjunctive Normal Form to Conjunctive Normal Form}
Similar to Lemma~\ref{lemma:max-monotone-cnf-to-dnf} in
Section~\ref{section:hardness}, the goal of this section is to
companion Corollary~\ref{corollary:abs-dnf} with a
\textsc{cnf}-version of the claim.

\begin{lemma}\label{lemma:abs-cnf}
  $\PLang[\alpha,d]{abs-cnf}\leq_{\text{fpt}} \PLang[\alpha,d]{abs-dnf}$
\end{lemma}
\begin{proof}
  Let $(\phi,w,\alpha)$ be an instance of $\PLang[\alpha,d]{abs-cnf}$
  and let $\Phi=\clause(\phi)$. The idea of the reduction is to
  successively replace every disjunction in $\Phi$ by a set of
  conjunctions with new weights such that we have for the resulting
  pair $(\Phi',w')$:
  \[
     \sum_{\sigma\in\Phi,  \beta\models\sigma}  w(\sigma)
   = \sum_{\sigma\in\Phi', \beta\models\sigma} w'(\sigma).
  \]

  Let $c$ be some disjunction in $\Phi$. Recall that $c$ contains at
  most $d$ variables. There is a set $X$ of at most $2^{|\var(c)|}$
  conjunctions such that we have for every truth assignment
  $\beta\colon\var(\phi)\to\{\mathit{true},\mathit{false}\}$ (i)
  $\beta\models c$ iff $\beta$ satisfies exactly one formula in $X$,
  and (ii) $\beta\not\models c$ iff $\beta$ satisfies no formula in
  $X$.  Such a set $X$ can easily be computed using the truth table
  of~$c$, which is possible in the reduction as $d$ is a
  parameter. Let $w'(c')=w(c)$ for any $c'\in X$, and
  $w'(\tilde c)=w(\tilde c)$ for any $\tilde c\in\Phi\setminus \{c\}$.

  For the reduction, we replace $\Phi$ by
  $\Phi'=\Phi\setminus\{c\}\cup X$ and update $w$ to $w'$ accordingly
  as long as there is a disjunction $c\in\Phi$. Note that this will
  replace exactly $|\clause(\phi)|$ disjunctions with at most $2^d$
  conjunctions each, which implies that the reduction can be carried
  out by an fpt-algorithm. The resulting pair $(\Phi', w')$ satisfies
  for every truth assignment~$\beta$:
  \begin{align*}
    \phantom{=}\quad &\sum_{\mathclap{\sigma\in\Phi, \beta\models\sigma}}                 \hspace{1ex} w(\sigma)   \\
    =\qquad &\sum_{\mathclap{\sigma\in\Phi\setminus\{c\}, \beta\models\sigma}}              
\hspace{1ex} w(\sigma)   &+&\quad &&\begin{cases} w(c)&\text{if $\beta\models 
c$;}\\0&\text{otherwise;}\end{cases}\hspace{5cm} \\
    =\qquad &\sum_{\mathclap{\sigma\in\Phi\setminus\{c\}, \beta\models\sigma}}              \hspace{1ex} w(\sigma)   &+&\quad &&\sum_{\mathclap{\sigma\in X, \beta\models c}}\hspace{1ex} w(c) \\
    =\qquad &\sum_{\mathclap{\sigma\in\Phi\setminus\{c\}, \beta\models\sigma}}              \hspace{1ex} w'(\sigma)  &+&\quad &&\sum_{\mathclap{\sigma\in X, \beta\models c}}\hspace{1ex} w'(\sigma) \\
    =\qquad &\sum_{\mathclap{\sigma\in\Phi\setminus\{c\}\cup X, \beta\models\sigma}} \hspace{1ex} w'(\sigma).   \tag*{\qedhere}\\
  \end{align*}
\end{proof}
\begin{corollary}
  $\PLang[\alpha,d]{abs-cnf}$, $\PLang[\alpha,d]{abs-monotone-cnf}\in \Class{FPT}$
\end{corollary}

\section{Application: Absolute Integer Optimization}\label{section:absip}

We have seen that $\PLang[\alpha,d]{abs-monotone-dnf}\in\Class{FPT}$.
We can generalize this problem to an algebraic optimization problem:
Given a sum of products (a \emph{multivariate polynomial} over binary
variables), find an assignment such that the absolute value of the sum
is at least $\alpha$. Using binary variables, this problem is
essentially the same problem as $\Lang{abs-monotone-dnf}$ (in fact,
optimization over \textsc{DNF}s is sometimes called sum of
products). However, what happens if we do not have binary variables,
but arbitrary integers from some given domain? We prove in this
section that the problem remains tractable for similiar parameters.

Let us represent a multivariate polynomial $p_{A,w}(z_1,\dots,z_n)$
over $n$ variables $z_1,\dots, z_n$ by an $n\times m$ matrix~$A$ and a
vector $w$ of length~$m$ as follows: Define $A\langle z,i,j\rangle$ to
be $z^{A_{i,j}}$ if $z\neq 0$ or $A_{i,j}\neq 0$ and $1$ otherwise;
then write
$p_{A,w}(z_1,\dots,z_n)=\sum_{j=1}^m w_j\cdot\prod_{i=1}^n A\langle
z_i,i,j\rangle$.

\begin{problem}[{$\PLang[\alpha,d]{abs-io}$}]
  \begin{parameterizedproblem}
    \instance An $n\times m$ matrix~$A\in \mathbb{N}^{n\times m}_{\geq 0}$,
    a weight vector $w\in\mathbb{Z}^m$,
    two bounding vectors ${b_\textit{min}}  \in(\mathbb{Z}\cup\{-\infty\})^{n}$ and  
$b_\textit{max}\in(\mathbb{Z}\cup\{\infty\})^{n}$,
    and a target value $\alpha\in\mathbb{N}$.
    \parameter $\alpha$, $d=\max_{j\in\{1,\dots,m\}}\sum_{i=1}^n A_{i,j}$
    \question Is there a solution $x=(x_1, x_2, \dots,x_n)\in\mathbb{Z}^n$ such that $|p_{A,w}(x)|\geq\alpha$ and
    for all $i\in\{1,\dots,n\}$ we have $(b_{\textit{min}})_{i}\leq x_i\leq (b_{\textit{max}})_{i}$?
  \end{parameterizedproblem}
\end{problem}

Here the parameter~$d$ is the maximal number of variables that do
occur in any product of the polynomial. If there is no exponent larger
than $1$, then $d$ is the number of variables in every product or
\emph{monomial}.

In this problem we restrict the domain of each variable to an
interval. One might ask whether we could allow arbitrary linear
inequations instead of upper and lower bounds, which is common in
linear programming. However, in this case the problem becomes
$\Class{W[1]}$-hard:

\begin{lemma}\label{lemma:absip-with-linear-eq}
  $\PLang[\alpha,d]{abs-io}$ with linear inequation constraints is $\Class{W[1]}$-hard.
\end{lemma}
\begin{proof}
  To prove this we turn an instance of $\PLang[k]{independent-set}$
  into an instance of $\PLang[\alpha,d]{abs-io}$ with additional
  inequations. For a graph $G=(\{1,\dots,n\},E)$ and a number
  $k\in\mathbb{N}$, let~$A\in\mathbb{Z}^{n\times n}$ be the identity
  matrix, $w=(1,\dots,1)\in\mathbb{Z}^n$, and $\alpha=k$. Further set
  $b_\textit{min} = (0)^n$, $b_\textit{max} = (1)^n$, and let
  $x_u + x_v \leq 1$ for every edge~$\{u,v\}\in E$ be additional
  constraints.

  Since all weights in $w$ are positive, the absolute value is the
  exact value of the sum. Every size-$k$ independent has now a
  corresponding variable assignment and vice versa.
\end{proof}

Observe that variables may have a huge or possibly infinite domain of
possible values and, hence, it is not obvious that
$\PLang[\alpha,d]{abs-io}$ is in $\Class{FPT}$.  However, we already
mentioned that $\PLang[\alpha,d]{abs-monotone-dnf}$ is equivalent to
$\PLang[\alpha,d]{abs-io}$ if (i) the domain of all variables is
$\{0,1\}$ (i.\,e., $b_\textit{min} = (0)^n$ and
$b_\textit{max} = (1)^n$), and if (ii) we have no exponents (i.\,e.,
$A\in\{0,1\}^{m\times n}$).

Therefore, we only have to consider the remaining cases. If (i) holds
but (ii) does not, we can simply replace all nonzero-values in $A$ by
$1$. This does not change the value of~$p_{A,w}(x_1,\dots,x_n)$ as by
(i) we have only binary variables and $w$ remains the same.

If (i) does not hold, we use a set of reduction rules to transform an
instance over an arbitrary domain into one in which the domain is a
superset of $\{0,1\}$. This is, of course, not the same as property
(i). However, while adapting the domain, we can transform the instance
such that it is reducible to an instance of
$\PLang[\alpha,d]{unbalanced-subgraph}$ with the following property:
Rule~\ref{rule:subedge-rule} either identifies it as a yes-instance or
does nothing. This property can be used by an algorithm for
$\PLang[\alpha,d]{abs-io}$ with the following trick: First modify the
domain such that it contains $\{0,1\}$; then virtually shrink the
domain to $\{0,1\}$ (this restricts the solution space and, thus, may
turn a yes-instance to a no-instance, but may not turn a no-instance
to a yes-instance); perform the reduction and apply
Rule~\ref{rule:subedge-rule}; either deduce that we are dealing with a
yes-instance (trivial decision of Rule~\ref{rule:subedge-rule}) or
that the instance is small (as the rule did nothing). In the latter
case restore the larger domain and explore a search tree to solve the
problem (the size of the search tree is bounded, as
Rule~\ref{rule:subedge-rule} did not trigger).

Let us first describe the reduction rules.  As before, we assume that
a rule may only be applied if rules with smaller numbers were applied
exhaustively. The first rule removes unnecessary variables, products,
and constraints; or detects that the instance has no solution.

\begin{reductionrule}\label{rule:trivial-vars-clauses-constraints}
  Apply the following modifications if possible:
  \begin{enumerate}
    \item If $w_j=0$ for some $j$, then remove the $j$-th column from $A$ and $w$.
    \item If there is some $i$ such that $A_{i,j}=0$ for every $j$,
      then remove the $i$-th row from $A$, $(b_{\textit{min}})_{i}$, and $(b_{\textit{max}})_{i}$.
    \item If there is some $i$ with $(b_{\text{max}})_{i}<(b_{\textit{min}})_{i}$, then return a trivial no-instance.
    \item 
    If there is some $i$ with $(b_{\text{max}})_{i}=(b_{\textit{min}})_{i}$, then for every $j$ where $A_{i,j}>0$ replace $w_j$ by $w_j\cdot (b_{\text{max}})_{i}^{A_{i,j}}$ and set $A_{i,j}=0$.
    \item If there are two equal columns $j_0,j_1$ in $A$, then replace $w_{j_0}$ by $w_{j_0}+w_{j_1}$ 
and remove the $j_1$th column from $A$ and $w$.
  \end{enumerate}
\end{reductionrule}

\begin{lemma}\label{lemma:trivials-vars-rule}
  Rule~\ref{rule:trivial-vars-clauses-constraints} is safe.
\end{lemma}
\begin{proof}
  For Item~1 and 2 observe that neither variables that do
  not occur in the represented polynomial nor products that contain
  zero as a factor contribute to the value of the polynomial. Item~3
  follows from the fact that there cannot be a solution if the
  domain of a variable is the empty set. Finally, Item~4 follows by
  the distributivity of addition and multiplication.
\end{proof}
If this rule cannot be applied, the domain of any variable contains at
least two consecutive values. With the next rule we change
the domains of all variables such that the domain of every variable
contains $0$ and $1$.

\begin{reductionrule}\label{rule:shift-domain}
  Apply the following modifications if possible:
  \begin{enumerate}
    \item If there is an $i$ with $\{0,1\}\not
    \subseteq[(b_{\textit{min}})_{i},(b_{\textit{max}})_{i}]$ and  $(b_{\textit{min}})_{i}\in \mathbb{Z}$,
    then for every column $j_0$ where $A_{i,j_0}>0$ add $c:=A_{i,j_0}$
    new $j_1,j_2,\dots,j_c$-th columns such that 
    \[
      A_{i',j_k} =
      \begin{cases}
        A_{i',j_0}  &\text{if $i'\neq i$}\\
        A_{i,j_0}-k &\text{if $i'= i$}
      \end{cases}
      \quad\text{and}\quad
      w_{j_k} = \binom{c}{k}\cdot w_{j_0}\cdot (b_{\textit{min}})_{i}^{k}\quad\text{for $k\in\{1,\dots,c\}$}.
    \]
    Finally set $(b_{\textit{min}})_{i}$ to $0$ and $(b_{\textit{max}})_{i}$ to
    $(b_{\textit{max}})_{i}-(b_{\textit{min}})_{i}$.
    \item If there is an $i$ with $(b_{\textit{min}})_{i}=-\infty$ and
      $(b_{\textit{max}})_{i}\leq 0$, then replace $w_j$ by $-w_j$ for every $j$ where
      $A_{i,j}$ is odd, and set $(b_{\textit{min}})_{i}$ to  $-(b_{\textit{max}})_{i}$ and 
      $(b_{\textit{max}})_{i}$ to $\infty$.
  \end{enumerate}
\end{reductionrule}

\begin{lemma}\label{lemma:rule-shift-domain}
  Rule~\ref{rule:shift-domain} is safe.
\end{lemma}
\begin{proof}
  Any valid solution~$(y_1,\dots,y_n)\in\mathbb{Z}^n$ for
  the output instance~$(A',w',b'_{\textit{min}},b'_{\textit{max}},\alpha)$ can be transformed
  into a solution~$(x_1,\dots,x_n)\in\mathbb{Z}^n$ for the
  input~$(A,w,b_{\textit{min}},b_{\textit{max}},\alpha)$ and vice versa:
  \begin{enumerate}
  \item Let $(x_1,\dots,x_n),(y_1,\dots,y_n)\in\mathbb{Z}^n$ such that
    $y_j=x_j$ for $i\neq j$ and $y_i=x_i-(b_{\textit{min}})_{i}$. Since the
    reduction rule only changes the bounding vectors at position $i$, we
    have for every $j\neq i$ that
    \[
      (b'_{\textit{min}})_{j}=(b_{\textit{min}})_{j}\leq x_j=y_j\leq
      (b_{\textit{max}})_{j}=(b'_{\textit{max}})_{j}.
    \]
    Consider $i$, then
    $(b_{\textit{min}})_{i}\leq x_i\leq (b_{\textit{max}})_{i}$ if,
    and only if:
    \[
      (b'_{\textit{min}})_{i}=0=(b_{\textit{min}})_{i}-(b_{\textit{min}})_{i}\leq x_i-(b_{\textit{min}})_{i}=y_i\leq
      (b_{\textit{max}})_{i}-(b_{\textit{min}})_{i}=(b'_{\textit{max}})_{i}.
    \]

    We show that
    $p_{A,w}(x_1,\dots,x_n)=p_{A',w'}(y_1,\dots,y_n)$. Recall that
    $x_j=y_j$ for all $j\neq i$. Since every column in the matrix
    relates to a product in the represented polynomial, we only need
    to consider those columns (product terms in the sum respectively)
    that were used by the rule. Therefore, it is sufficient to show
    that for every $j_0$-th column in $A$ where $A_{i,j_0}>0$ the
    following holds:
    \[
      w_{j_0} \prod_{r=1}^n A\langle  x_r,r,j \rangle = \sum_{k=0}^{A_{i,j_0}} w'_{j_k} \prod_{r=1}^n A'\langle  y_r,r,j \rangle.
    \]
    This can be seen as follows:
    \begin{align*}
         &w_{j_0}                                     &\cdot&\quad &&\prod_{r=1}^n A\langle  x_r,r,j \rangle\hspace{6cm}\\
      =\quad &w_{j_0} \cdot x_i^{A_{i,j_0}}              &\cdot&\quad  &&\prod_{\mathclap{\substack{r\in\{1,\dots,n\}\\r\neq i}}} A\langle  x_r,r,j \rangle\\
      =\quad &w_{j_0} \cdot (y_i+(b_{\textit{min}})_{i})^{A_{i,j_0}} &\cdot&\quad &&\prod_{\mathclap{\substack{r\in\{1,\dots,n\}\\r\neq i}}} A\langle  y_r,r,j \rangle\\
      =\quad &\sum_{k=0}^{A_{i,j_0}} \binom{A_{i,j_0}}{k} w_{j_0}
          \cdot (b_{\textit{min}})_{i}^k \cdot y_i^{A_{i,j_0}-k}     &\cdot&\quad &&\prod_{\mathclap{\substack{r\in\{1,\dots,n\}\\r\neq i}}} A\langle  y_r,r,j \rangle\\
      =\quad &\sum_{k=0}^{A_{i,j_0}} w'_{j_k}              &\cdot&\quad &&\prod_{r=1}^n A'\langle  y_r,r,j \rangle.\\
    \end{align*}
  \item Let $(x_1,\dots,x_n),(y_1,\dots,y_n)\in\mathbb{Z}^n$  such
    that $y_j=x_j$ for $i\neq j$ and $y_i=-x_i$. The claim follows by
    the same argument using the substitution $x_i=-y_i$ in
  $p_{A,w}(x_1,\dots,x_n)$.\qedhere
  \end{enumerate}
\end{proof}

The first subrule shifts the domain of variables that have a finite
lower bound. Variables without such a bound are turned into variables
with a finite lower bound by the second subrule. The first subrule may
then be applied again.

As soon as none of the rules can be applied, every variable that
occurs in the polynomial~$p_{A,w}$ has a domain that is a superset
of~$\{0,1\}$. We can turn our reduced instance into an instance of
$\PLang[\alpha,d]{abs-dnf}$, which directly translates into an
equivalent instance of $\PLang[\alpha,d]{unbalanced-subgraph}$.

Thereby we might turn an $\Lang{abs-io}$ yes-instance into a
no-instance of \Lang{unbalanced-subgraph}. Recall that running the
kernelization algorithm from Section~\ref{section:absdnf} (especially
Rule~\ref{rule:subedge-rule}) either returns a trivial yes-instance or
does nothing. After applying
Rule~\ref{rule:trivial-vars-clauses-constraints} on the initial
instance, rules~\ref{rule:isolated-vertices}
and~\ref{rule:edges-without-weights} cannot be applied on the
hypergraph.  Therefore, the inital instance is a yes-instance or the
size of the matrix is bounded by the parameters (because $(H,w)$ is a
kernel).  Note that this does not imply that we obtain a kernel: The
values in the bounding vectors are eventually not bounded yet.
However, we show that such problematic instances can be reduced to a
set of equivalent instances in which we have control over all these
values.

\absip*
\begin{proof}
  Let us firstly present some branching rule for our problem:
  
  \begin{branchingrule}\label{rule:branching}
    Let there be some~$i^*$ where $(b_{\textit{max}})_{i^*}-(b_{\textit{min}})_{i^*}\geq
    2e\alpha$ with $e=\max_j A_{i^*,j}$. Let $w^{(k)}\in\mathbb{Z}^m$
    with $w^{(k)}_j=w_j$ if $A_{i^*,j}=k$ and $w^{(k)}_j=0$ otherwise,
    and let $A'$ be the matrix obtained by setting all values in the $i^*$-th
    row of $A$ to $0$.
    Then
    $(A,w,b_{\textit{min}},b_{\textit{max}},\alpha)\in\Lang{abs-io}$ iff
    there is some~$k\in\{0,\dots,e\}$ such that
    $(A',w^{(k)},b_{\textit{min}},b_{\textit{max}},\alpha_k)\in\Lang{abs-io}$ where $\alpha_k=1$ for $k > 1$ and $\alpha_0=\alpha$.
  \end{branchingrule}
  \begin{claim}\label{lemma:branching-rule}
    Branching Rule~\ref{rule:branching} is safe.
  \end{claim}
  \begin{proof}
    Let $(x_1, \dots, x_n)$ be an arbitrary solution for the input
    instance. By term rewriting it follows that
    \begin{align*}
      p_{A,w}(x_1,\dots,x_n) &= \sum_{j=1}^m w_j\cdot\prod_{i=1}^n A\langle x_{i},i,j\rangle\\
      &=  \sum_{k=0}^{e}\left(\sum_{
        \substack{j\in\{1,\dots,m\}\\A_{i^*,j}=k}
        } w_j\cdot\prod_{\substack{i\in\{1,\dots,n\}}} A\langle x_{i},i,j\rangle\right)\\
        &=  \sum_{k=0}^{e}\left(\sum_{
          \substack{j\in\{1,\dots,m\}\\A_{i^*,j}=k}
          } w^{(k)}_j\cdot\prod_{\substack{i\in\{1,\dots,n\}\\i\neq i^*}} A'\langle x_{i},i,j\rangle \cdot A\langle x_{i^*},i^*,j,\rangle\right)\\
        &=  \sum_{k=0}^{e}\left(p_{A',w^{(k)}}(x_1,\dots,x_n) \cdot A\langle x_{i^*},i^*,k,\rangle\right).\\
    \end{align*}
    We show that one of the branching instance has a solution, too.  If
    $p_{A',w^{(k)}}(x_1,\dots,x_n)>0$ for some $k>0$, we are
    done. Otherwise $p_{A',w^{(k)}}(x_1,\dots,x_n)=0$ holds for every
    $k>0$. Then the equation above implies
    $p_{A',w^{(0)}}(x_1,\dots,x_n) = p_{A,w}(x_1,\dots,x_n)$. Since
    $|p_{A,w}(x_1,\dots,x_n)|\geq \alpha$, we get that
    $|p_{A',w^{(0)}}(x_1,\dots,x_n)|\geq\alpha$, which means
    $(A',w^{(0)},b_{\textit{min}},b_{\textit{max}},\alpha)\in\Lang{abs-io}$.

    It remains to show that the input instance is a yes-instance if
    one of the branching instances is a yes-instance. First consider
    $(A',w^{(0)},b_{\textit{min}},b_{\textit{max}},\alpha)$. Note,
    that $p_{A',w^{(0)}}(x_1,\dots,x_n)$ does not depend on $x_{i^*}$
    because $w^{(0)}_j=0$ for every product that includes
    $x_{i_*}$. For the same reason we have:
    \[
      p_{A,w}(x_1,\dots,x_{i^*-1},0,x_{i^*+1},\dots,x_n) =p_{A',w^{(0)}}(x_1,\dots,x_n).
    \]
    Hence, we can turn a
    solution of
    $(A',w^{(0)},b_{\textit{min}},b_{\textit{max}},\alpha)$ into a
    solution for the input
    instance~$(A,w,b_{\textit{min}},b_{\textit{max}},\alpha)$ by
    replacing~$x_i$ by~$0$.
    
    Now assume that
    $(A',w^{(0)},b_{\textit{min}},b_{\textit{max}}),\alpha)$ is a
    no-instance and consider the largest number~$k\in\{1,\dots,e\}$
    such that $(A',w^{(k)},b_{\textit{min}},b_{\textit{max}},1)$ is a
    yes-instance. Let $(y_1,\dots,y_n)$ be one of its
    solutions. Define the
    polynomial~$p(z):=p_{A,w}(y_1,\dots,y_{i^*-1},z,y_{i^*+1},\dots,y_n)$. It
    follows that $p(z) = \sum_{\ell=0}^{k}c_\ell \cdot z^\ell$ with
    $c_\ell = p_{A',w^{(\ell)}}(y_1,\dots,y_{n})$. Note that for every
    $z$ we have
    $|p_{A',w^{(k)}}(y_1,\dots,y_{i^*-1},z,y_{i^*+1},\dots,y_n)|\geq
    1$ as $A'_{i^*,j}=0$ for every column~$j$ and, thus, $c_k\neq 0$.
    We will now argue by curve sketching over $\mathbb{R}$ that there
    is a
    $z\in\{(b_{\textit{min}})_{i},\dots,(b_{\textit{max}})_{i}\}\cap\mathbb{Z}$
    with $|p(z)|\geq \alpha$. Any interval~$[a,b]$ with $b-a>2\alpha$
    where either $(a,b]$ or $[a,b)$ contains no extremum of the
    polynomial contains at least one value $z\in\mathbb{Z}$ where
    $|p(z)|\geq \alpha$. This is because the polynomial is strictly
    increasing or decreasing on such an interval and the domain and
    image are restricted to $\mathbb{Z}$. The polynomial $p$ has a
    degree of $k$, which is at most~$e$, and, thus, it has at most
    $e-1$ extreme points. Since
    $(b_{\textit{max}})_{i}-(b_{\textit{min}})_{i}\geq 2e\alpha$,
    there is such an interval within
    $[(b_{\textit{min}})_{i};(b_{\textit{max}})_{i}]$. Hence, there is
    a $z\in[(b_{\textit{min}})_{i};(b_{\textit{max}})_{i}]$ with
    $|p(z)|\geq \alpha$. By the definition of $p(z)$, we finally get
    that $(y_1,\dots,y_{i-1},z,y_{i+1}\dots,y_n)$ is a solution for
    the input instance.
  \end{proof}

  Note that this rule reduces the number of variables on which the
  value of the polynomial depends. Recall
  Rule~\ref{rule:trivial-vars-clauses-constraints}: As soon as a row
  contains only zeroes, it can be safely removed. Since the rules
  output at most~$e\leq d$ branches, an exhaustive application of 
  all the reduction rule results in a search tree of size
  $O(d^n)$. This search tree computes $O(d^n)$ instances in total, all
  of which have no variables with large domain (more than $2d\alpha$
  possible values).
  
  We partition the remaining part of the proof into three parts:
  First, we present the algorithm sketched in the main text in detail;
  second, we prove that this algorithm correctly solves
  $\PLang[\alpha,d]{abs-io}$; and third, we argue that the
  algorithm in fpt-time.
  
  \subparagraph*{The Algorithm.}  Our algorithm has four phases. In
  the first phase it exhaustively applies the reduction
  rules~\ref{rule:trivial-vars-clauses-constraints}
  and~\ref{rule:shift-domain} to obtain
  an instance~$(A,w,b_{\textit{min}},b_{\textit{max}},\alpha)$. If
  Rule~\ref{rule:trivial-vars-clauses-constraints} returns a trivial
  no-instance, we reject the initial input. As soon as no
  reduction rule can be applied, it follows from
  Rule~\ref{rule:trivial-vars-clauses-constraints} that every variable
  in $p_{A,w}$ has a domain of size at least two. By
  Rule~\ref{rule:shift-domain} every domain must be a superset of
  $\{0,1\}$.

  In the second phase we use the kernelization algorithm for
  $\PLang[\alpha,d]{unbalanced-subgraph}$: For the instance
  $(A,w,b_{\textit{min}},b_{\textit{max}},\alpha)$ obtained from
  the previous phase, let $H=(V,E)$ be a $d$-hypergraph and let
  $w\colon E\to\mathbb{Z}$ be a weight function with
  $V=\{1,\dots,n\}$, $E=\{e_1,\dots,e_m\}$, 
  $e_j=\{i \mid A_{i,j}>0\}$, and $w(e_j) = w_j$. We test whether
  $|E|\geq \edgenumber{}$ and accept if so.

  Otherwise $m=|E|\leq \edgenumber{}$ and $n\leq d\cdot m$. Then
  we continue with the third phase and exhaustively apply Branching
  Rule~\ref{rule:branching} as well as the reduction rules
  (Rule~\ref{rule:trivial-vars-clauses-constraints} and
  Rule~\ref{rule:shift-domain}). This gives us a search tree that
  computes a set of instance on which neither a reduction rule
  nor the branching rule is applicable. The domain of all variables in
  all of these instances is bounded by the parameters and, hence, we
  can solve them via ``brute-force.''

  \subparagraph*{Proof of Correctness.}
  Consider the first phase of the algorithm where
  rules~\ref{rule:trivial-vars-clauses-constraints} and
  \ref{rule:shift-domain} are applied in the given order as long as
  possible. Since the rules are safe, it is correct
  to stop and reject if
  Rule~\ref{rule:trivial-vars-clauses-constraints} returns a
  trivial no-instance. Otherwise we obtain an instance that has a
  solution if, and only if, the initial input instance does.

  Since the reduction rules have been applied exhaustively in the
  first phase, the domain for every variable is a superset of
  $\{0,1\}$ afterwards. Therefore, any solution for the constructed
  $\Lang{unbalanced-subgraph}$ instance gives us a valid solution. By
  Rule~\ref{rule:trivial-vars-clauses-constraints} there is no
  isolated vertex in $V$ nor an edge $e\in E$ with $w(e)=0$ in the
  constructed hypergraph. This means the only rules
  that modify the hypergraph (Rule~\ref{rule:isolated-vertices} and
  Rule~\ref{rule:edges-without-weights}) are
  not applicable.  Because of that, from the proof of
  Corollary~\ref{corollary:kernel-parameters-alpha-d}, and from
  Lemma~\ref{lemma:rule-four-safe}, the test for the size of the
  hypergraph correctly identifies trivial yes-instances or does
  nothing. We can derive a solution~$(x_1,\dots,x_n)\in\mathbb{Z}^n$
  with $|p_{A,w}(x_1,\dots,x_n)|\geq\alpha$ from a solution~$X$ for
  our hypergraph where $|w[X]|\geq\alpha$ as follows by setting $x_i = 1$ if
  $i\in X$, and by setting $x_i=0$ otherwise. This is correct as the
  term that we maximize for the $\Lang{unbalanced-subgraph}$ instance is
  the same as $|p_{A,w}(x_1,\dots,x_n)|$, and because the domain of
  every variable contains $0$ and $1$.
  
  In the third phase, the algorithm reduces the instance to a set of
  instances in which the domain of every variable is bounded by the
  parameter.  Variables with a too large domain will be eliminated by
  the branching rule and a following application of
  Rule~\ref{rule:trivial-vars-clauses-constraints} to remove the row
  containing only zeros. From the safeness of the branching rule it
  directly follows that there is one yes-instance if, and only if, one
  of the final branching instances is a yes-instance. Since the number
  of possible solutions for each instance is finite, a brute-force algorithm will decide for each
  of these instances correctly whether it is a yes-instance. 

  \subparagraph*{Runtime Analysis.}
  We first argue that applying the
  reduction rules exhaustively can be done in time
  $f(\alpha,d)\cdot \mu^c$, were $\mu$ is the encoding length of the
  input, $f$ some computable function, and $c\in\mathbb{N}$ a
  constant.
 
  Rule~\ref{rule:shift-domain} removes a variable with a domain that
  is not a superset of $\{0,1\}$ and no rule introduces new
  variables. Hence, Rule~\ref{rule:shift-domain} is applied at most
  $O(n)$ times.  Rule~\ref{rule:trivial-vars-clauses-constraints} can
  only be applied $O(n+m)$ times because subrules 1--4 can be applied
  only once initially; the fifth subrule can be applied after
  Rule~\ref{rule:shift-domain}. Hence the rules can only
  be applied a polynomial number times. Observe that almost all
  reduction rules are computable in polynomial time with respect to
  the size of the input instance.  The sole exceptions are
  Rule~\ref{rule:shift-domain}.1 and
  Rule~\ref{rule:trivial-vars-clauses-constraints}.4, which still run
  in fpt-time.

  We also have to show that the size of subproblems does not exceed
  the algorithm's time bound of $f(\alpha,d)\cdot \mu^c$. The critical
  part is Rule~\ref{rule:shift-domain}.1~--~in all other cases the
  instance becomes only smaller. Each time we apply the rule for some
  $i$, the number of columns~$j$ for which $A_{i,j}>0$ is increased by
  a factor of at most $d$ as $A_{i,j}\leq d$. Every new column arises
  directly or indirectly from one column that was in the initial
  instance and, hence, the size of the instance increases at most by a
  factor of $d^d$ in total.
 
  For the search tree note that once the kernelization was used in the
  second phase, $n$ and $m$ are bounded by the parameters. The height
  and size of the search tree is therefore bounded, too. Note that the
  kernelization itself runs in polynomial time by the definition of a
  kernelization algorithm.  Finally, whenever Branching
  Rule~\ref{rule:branching} cannot be applied, we have to solve the
  instance at the leaves of the search tree.  Since neither a reduction
  rule nor the branching rule are applicable, it holds that $n$, $m$,
  all values in the matrix, and all values in the bounding vectors are
  bounded by the parameters for these instances. Thus, the brute-force
  algorithm's runtime is bounded by some computable function that
  depends only on the parameters (note that the weights are not
  necessarily bounded, but this is no problem for a combinatorial
  algorithm that enumerates the solutions).
\end{proof}

\section{Conclusion and Outlook}\label{section:conclusion}

We considered several variants of the maximum satisfiability problem
and showed that, as soon as we allow negative weights, those variants
become $\Class{W[1]}$-hard parameterized by the solution
size~$\alpha$~--~even for monotone clauses of fixed size~$d$.  On the
other hand, we obtained fixed-parameter tractability results
parameterized by $\alpha+d$ if we optimize the absolute value of the
target function.  The latter result was obtained via a kernelization
for the auxiliary hypergraph problem
$\PLang[\alpha,d]{unbalanced-subgraphs}$, which we think may be of
independent interest. Using these techniques, we were able to almost
completely resolve the complexity of $\Lang{abs-dnf}$, see
Table~\ref{table:results} in the introduction for an overview. The
only remaining case is $\Lang{$\alpha$-abs-dnf}$, i.\,e., the version
in which the sought solution size $\alpha$ is constant while the size
of the clauses is unbounded. We do not see how our techniques can be
used for this version, as our algorithms for the unbalanced subgraph
problem rely on hyperedges of bounded size.

Using a collection of additional reduction rules, we were able to
generalize the results from $\PLang[\alpha,d]{abs-dnf}$ to
$\PLang[\alpha,d]{abs-io}$, which tries to optimize the absolute value
of the target function of a restricted integer optimization problem.

An interesting line of further research could be to study the
\emph{minimization} version of the problems presented within this
paper. Usually, minimization and maximization problems have similar
complexity, as one can perform some easy modifications such as
multiplying all weights with $-1$. This is, however, not the case if
we optimize the absolute value of the target function, as the
following observation illustrates: Let
$\Lang{$d$-min-abs-monotone-dnf}$ be defined as
$\Lang{$d$-abs-monotone-dnf}$, but with $\geq\alpha$ being replaced by
$\leq\alpha$, then:

\begin{observation}
  $\Lang{independent-set}\leq\Lang{$d$-min-abs-monotone-dnf}$
\end{observation}
\begin{proof}[Sketch of Proof.]
  We use the reduction from Lemma~\ref{lemma:max-dnf} to reduce
  \Lang{independent-set} to \Lang{$d$-max-monotone-dnf}. Hence, we
  obtain weighted formula $(\phi,w)$ and seek an assignment of weight
  at least $\geq\alpha$ (without absolute values). Note that if such
  an assignment exists, then there is also one with weight $=\alpha$
  as the independent set problem is monotone. We continue by adding
  the empty clause (which is always true in a \textsc{dnf}) and set
  its weight to~$-\alpha$. Finally, we seek an solution with absolute
  value $\leq 0$.
\end{proof}

Note that this reduction works for constant $d$ and produces an
instance with $\alpha=0$, i.\,e., a parameterization by $\alpha$ and
$d$ alone is not enough. However, restricting, for instance, the
weights may yield tractable subproblems that should be explored
further.

\clearpage
\bibliography{main}
\end{document}